\documentclass{article}
\usepackage[letterpaper, portrait, margin=1in]{geometry}
\usepackage{graphicx} % Required for inserting images

\usepackage{authblk}

\usepackage[ruled]{algorithm2e}

\usepackage{xcolor}
\definecolor{teal}{HTML}{008080}

\usepackage{pgfplots}
\usepackage{comment}

\usepackage{hyperref}
\hypersetup{
   colorlinks=true,
   citecolor=teal   
}

\usepackage{amsmath}
\usepackage{amssymb, amsthm}

\usepackage{algorithm2e} % for algorithms

\newtheorem{theorem}{Theorem}
\newtheorem{lemma}{Lemma}
\newtheorem{corollary}{Corollary}

\theoremstyle{definition}
\newtheorem{example}{Example}
\newtheorem{definition}{Definition}

\newcommand{\client}{\mathcal{C}}
\newcommand{\server}{\mathcal{S}}
\newcommand{\analyst}{\mathcal{T}}

\allowdisplaybreaks[1] % To allow multiline equations over multiple pages. For amssmath package.

\usepackage{tikz}

\usetikzlibrary{quantikz2}

\usetikzlibrary{chains,positioning}
\tikzset{>=latex}

\newcommand{\descr}[1]{\vspace{0.2cm} \noindent \textbf{#1}}

\usepackage[textsize=scriptsize, textwidth=2.1cm]{todonotes}
\setuptodonotes{color=green!40}

\title{Answering Counting Queries with Differential Privacy on a Quantum Computer}

\author[1]{Arghya Mukherjee}
\author[1]{Hassan Jameel Asghar}
\author[2]{Gavin K. Brennen}
\affil[1]{School of Computing, Macquarie University, Australia\\
\texttt{hassan.asghar@mq.edu.au},  \texttt{arghya.mukherjee@students.mq.edu.au}} 
\affil[2]{School of Mathematical and Physical Sciences, Macquarie University, Australia\\\texttt{gavin.brennen@mq.edu.au}}

\date{\today}

\begin{document}

\maketitle

\begin{abstract}
Differential privacy is a mathematical notion of data privacy that has fast become the de facto standard in privacy-preserving data analysis. Recently a lot of work has focused on differential privacy in the quantum setting. Continuing on this line of study, we investigate how to answer counting queries on a quantum encoded dataset with differential privacy. An example of a counting query is ``How many people in the dataset are over the age of 25 and with a university education?'' Counting queries form the most basic but nonetheless rich set of statistics extractable from a dataset. We show that answering these queries on a quantum encoded dataset reduces to measuring the amplitude of one of two orthogonal states. We then analyze the differential privacy properties of two algorithms from literature to measure amplitude: one which performs repeated measurements in the computational basis, and the other which utilizes the classic amplitude estimation algorithm. For the first technique, we prove privacy results for the case of counting queries that improve on previously known results on general queries, and show that the mechanism in fact \emph{amplifies} privacy due to inherent randomness. For the second method, we derive a tight bound on maximum possible change in the amplitude if we add or remove a single item in the dataset, a quantity called global sensitivity which is central in making an algorithm differentially private. We then show a differentially private version of the amplitude estimation algorithm for counting queries. We also discuss how these methods can be outsourced to a quantum server to blindly compute counting queries with differential privacy. 
\end{abstract}

\section{Introduction}
We investigate the problem of answering statistical queries with privacy on a classical dataset stored as a quantum encoded dataset on a quantum computer. Many uses cases require statistical analysis of sensitive datasets, such as the analysis of census data by demographics. The goal is to allow this analysis without compromising the privacy of individuals contributing their data. In the classical setting, the rigorous framework of differential privacy~\cite{dwork2006calibrating} gives a solution to this problem. Informally, differential privacy guarantees that the outcome of a differentially private algorithm, which is given the target dataset as input, is almost indistinguishable if any single person's data is added or removed. With a differentially private algorithm, the database owner can answer queries such as the number of people residing in a region with a particular demographic, without fear that the answers to these queries might result in a privacy breach.

We are interested in the scenario where a database owner $\client$, which we call the client, uploads its dataset to a quantum computing server $\server$ which in turn answers queries on the dataset from a third party called the analyst $\analyst$. Just as in the classical setting, the mechanism should answer queries via differential privacy. We are particularly interested in mechanisms that answer a certain class of queries called \emph{counting queries}. A subset of such queries include \emph{predicate queries}, which arise naturally in a variety of real-world datasets, such as census data, and have been analyzed extensively in classical differential privacy literature~\cite{mckenna2021hdmm, smith2022making}. An example of such a query is ``How many people in the dataset are students and within the age bracket of 25-34?'' With access to the classical dataset and classical differentially private mechanisms, such as the Laplace mechanism~\cite{dwork2006calibrating}, the server $\server$ can easily answer such queries. In more detail, the server can compute the original answer, and then add \emph{noise} of appropriate scale by sampling a random variable from the Laplace distribution. However, we are interested in the scenario where the dataset is encoded as quantum information. The main motivation behind this scenario is that recent work on quantum algorithms for differential privacy show that privacy is amplified if a mechanism takes the quantum encoding of the dataset as input, instead of the original, classical data~\cite{angrisani2022differential}. Here, privacy amplification means that the differentially private mechanism is in fact more private than what the rudimentary analysis suggests, owing to the inherent randomization of the process. In the classical setting, this is, for instance, true of subsampling, in which the differentially private algorithm is applied on random subsamples rather than the entire dataset~\cite{balle2018privacy}. Thus, our goal is a more accurate privacy loss analysis of algorithms that answer predicate (or counting) queries on a quantum encoded dataset.

Another motivating factor for this use case is the potential to efficiently outsource computation to the server $\server$ without the server learning anything about the dataset, albeit with some interaction with the client $\client$. This can be done by encrypting the quantum encoded dataset using the quantum one-time pad (QOTP) (see~\cite{vidick2023introduction}, for example), and the server $\server$ can then answer differentially private queries over the encrypted dataset using the homomorphic properties of the QOTP. In the classical setting, this can also be done using a fully homomorphic encryption (FHE) scheme. The client $\client$ uploads its encrypted dataset to the server $\server$, who then allows the execution of differentially private queries over the encrypted dataset using the homomorphic properties of the encryption scheme. The final result can then be decrypted by the client and sent to the analyst $\analyst$.\footnote{Alternatively, we could use secure multiparty computation (SMC) techniques, with a trade-off of more communication between $\client$ and $\server$.}  
% by the $\client$ encrypting the dataset $D$ using a fully homomorphic encryption (FHE) scheme $\mathcal{E}$ as $\mathcal{E}(D)$, and sending the encrypted dataset to $\server$. A query $q$ from the $\analyst$ can then be computed on the encrypted dataset $\mathcal{E}(D)$ as $q(\mathcal{E}(D)) = \mathcal{E}(q(D))$ using the homomorphic properties of $\mathcal{E}$ and then either $\server$ sends the encrypted answer $\mathcal{E}(q(D))$ to $\client$, who can decrypt it and then add additive differentially private noise before sending the final answer to $\analyst$, or $\server$ adds differentially private noise to $\mathcal{E}(q(D))$ using the homomorphic properties of $\mathcal{E}$ again, before sending it to $\client$ who simply decrypts the answer and sends the result to $\analyst$. As long as $\server$ is honest-but-curious and there is no collusion between $\server$ and $\analyst$ the scheme achieves the desired confidentiality, with respect to $\server$, and privacy, with respect to both $\analyst$ and $\server$, goals.\footnote{Alternatively, we could use secure multiparty computation (SMC) techniques, with a trade-off that of more communication between $\client$ and $\server$.}  
However, the issue with the classical approach is that FHE operations on a large encrypted dataset are computationally expensive in practice. The QOTP scheme is potentially more efficient since encryption can be done via simple Pauli gates, and most of the homomorphic computation using Pauli or Clifford gates can be carried out over the encrypted data with some change to the encryption keys requiring client interaction. Moreover, non-Clifford gates can also be performed, although with extra machinery and a quantum client (see Section~\ref{sec:homomorphic} as well as~\cite{vidick2023introduction}).

Motivated by these observations, this paper explores this setting in detail and makes the following contributions.

\begin{itemize}
    \item We show that since a counting query can be formulated as a Boolean function, its effect on the quantum encoded state is to decompose it into two orthogonal components, one of which exactly corresponds to the rows of the dataset that evaluate to true under the query (modeled as a predicate). In particular, the probability of measuring this state exactly equals the answer to the counting query. We then discuss two methods to extract the query answer. The first estimates the (square of the) amplitude via repeated measurements as described in~\cite{angrisani2022differential}, and the second uses the amplitude estimation method from~\cite{Brassard_2002}. Both methods have not been used in the context of answering differentially private counting queries. The first method is described in~\cite{angrisani2022differential} for a generic measurement operator with binary outcomes, whereas the second method's main use is in searching or counting solutions to a Boolean function in a database. 

    \item For the first method, i.e., repeated direct measurements, we first show that using the algorithm and privacy amplification results shown in~\cite{angrisani2022differential}, we in fact end up adding more noise than differential privacy in the classical setting. We instead perform a fresh analysis of the differential privacy properties of this algorithm, and show that the algorithm allows a spectrum of differential privacy guarantees for counting queries, including one parameter setting where we obtain differential privacy without adding explicit noise! Our results show that less noise needs to be added to the result of a counting query to make it differentially private than previously shown for a general case~\cite{angrisani2022differential}. A distinguishing characteristic of a counting query over a generic query is that its global sensitivity is 1. Roughly, global sensitivity of a function is the maximum change in its value if we add or remove a single data point. 

    \item For the second method, we add noise of appropriate scale to the phase of the eigenvalue of the main operator used in the amplitude estimation method from~\cite{Brassard_2002} to make the algorithm differentially private (see Section~\ref{subsec:amp-est}). To do so, our main contribution is to derive the global sensitivity of the phase of this eigenvalue, through the global sensitivity of a counting query. We then devise a differentially private version of the amplitude estimation algorithm for counting queries by defining a unitary operator that adds noise of appropriate scale to the phase of the eigenvalue.

    \item We briefly discuss how inherent quantum noise, modeled as a depolarizing channel after every unitary, can be factored into our results to give an overall differential privacy guarantee: from state preparation to query evaluation. We further discuss how the mechanisms, including the inputs and outputs, can be computed blindly by the server using the quantum one-time pad along with its homomorphic properties.
\end{itemize}

\section{Background}
We give a brief background on differential privacy and how datasets can be encoded as quantum information, together with known results connecting the two.  

\subsection{Differential Privacy}

Let $\mathcal{X}$ be a data domain. We assume a dataset to be a tuple of $n$  elements $D = (x_1, \ldots, x_n) \in \mathcal{X}^n$, where each $x_i \in \mathcal{X}$. Two datasets $D, D'$ are neighbours, denoted $D \sim D'$ if they differ in exactly one element. %Let $\epsilon \ge 0$, and $\delta \in [0, 1]$. 

\begin{definition}[Differential Privacy~\cite{dwork2006calibrating}]
\label{def:dp}
A randomized algorithm $\mathcal{M}: \mathcal{X}^n \rightarrow \mathcal{Y}$ is $(\epsilon, \delta)$-differentially private (DP) if for all pairs of neighbouring datasets $D, D'$ and for all subsets $S \subseteq \mathcal{Y}$, we have:
\[
\Pr[\mathcal{M}(D) \in S] \leq e^\epsilon \Pr[\mathcal{M}(D') \in S] + \delta,
\]
where $\epsilon > 0$ and $\delta$ is a negligible function of $n$. If $\delta = 0$, we say that the mechanism is $\epsilon$-differentially private. 
\end{definition}
The following are known results.
\begin{theorem}[Post-processing~\cite{dwork2006calibrating}]\label{prop:post-processing}
    If a mechanism $\mathcal{M}: \mathcal{X}^n \rightarrow \mathcal{Y}$ is $(\epsilon,\delta)$-DP, then for any function $f$ with domain $\mathcal{Y}$, the function $f \circ \mathcal{M}$ remains $(\epsilon,\delta)$-DP.
\end{theorem}

\begin{theorem}[Sequential composition~\cite{dwork-dp-book}]\label{prop:composition}
    If $\mathcal{M}_1,\ldots,\mathcal{M}_t$ are $(\epsilon,\delta)$-DP mechanisms, then the sequence of algorithms $\mathcal{M}' = \left(\mathcal{M}_1,\ldots,\mathcal{M}_t\right)$ is $(t\epsilon, t\delta)$-DP.
\end{theorem}

A \emph{query} $q$ is defined as a function from $q: x \rightarrow \mathbb{R}$, where $x$ is any row of the dataset. Overloading notation, the query $q$ on the dataset $D$ is defined as $q(D) \coloneqq \frac{1}{n}\sum_{i = 1}^n q(x_i)$.\footnote{We divide by $n$ so that all query answers are within the interval $[0, 1]$. One can obtain the corresponding unnormalized value by simply multiplying the result by $n$.}

\begin{definition}[Global Sensitivity]
The global sensitivity of a query $q$, denoted $\Delta q$, is defined as 
\[
\Delta q = \max_{D \sim D'} |q(D) - q(D')| 
\]
\end{definition}
Thus, global sensitivity captures the maximum possible change in the answer to the query over all possible neighboring datasets. 

\begin{definition}[Laplace mechanism~\cite{dwork2006calibrating}]
\label{def:laplace}
The zero-mean Laplace distribution has the probability density function $\text{Lap}(x \mid b) = \frac{1}{2b} e^{-\frac{|x|}{b}}$, where $b$, a non-negative real number, is a scale parameter. Given a dataset $D$ and query $q$, the Laplace mechanism is defined as 
\[
\mathcal{M}_{\text{Lap}}(q, D) = q(D) + \eta 
\]
where $\eta$ is a Laplace random variable of scale $\Delta q / \epsilon$.
\end{definition}

\begin{theorem}
The Laplace mechanism is $\epsilon$-differentially private~\cite{dwork2006calibrating}. 
\end{theorem}

We also recall the definition of quantum differential privacy (QDP):
\begin{definition}[Quantum Differential Privacy~\cite{zhou2017differential, angrisani2022differential}]
A quantum operation $\mathcal{E}$ is $(\tau, \varepsilon, \delta)$-quantum differentially private (QDP) if for every positive operator-valued measure (POVM) $M$, for all subsets $S$ of possible outcomes, and for all states $\rho$ and $\sigma$ such that $\mathcal{T}(\rho, \sigma) \leq \tau$, 
\[
\Pr\big[M(\mathcal{E}(\rho)) \in  S\big] 
\leq e^{\varepsilon} \Pr\big[M(\mathcal{E}(\sigma)) \in  S\big] + \delta.
\]
\end{definition}

See the next section for the definition of the trace distance $\mathcal{T}$, and quantum states $\rho$ and $\sigma$.

\subsection{Quantum Encoding of Datasets}
We call each element $x_i$ of the dataset $D = (x_1, \ldots, x_n)$ a \emph{row} of the dataset, and further assume that each $x_i \in \{0, 1\}^m$ for a fixed positive integer $m$. Typically, a dataset is formatted as a table with columns representing attributes and rows representing data from individuals. Based on the range of values taken by each attribute, discretizing if it is continuous, we can represent each value as a binary string of a fixed length. Thus, without loss of generality, we can represent each row of the dataset as a string in $\{0, 1\}^m$, where $m$ is determined by the attributes in the dataset. We can define a \emph{feature map} $\phi$ which encodes the dataset $D$ into its quantum encoding. The most common method is \emph{basis encoding} which encodes $D$ as 
\begin{equation}
\label{eq:state-prepare}
    \ket{\phi_D} = \frac{1}{\sqrt{n}}\sum_{i = 1}^n \ket{x_i}
\end{equation}
Here $\ket{x_i}$ represents an $m$-dimensional state. In case there are duplicates in the dataset, which would result in an unnormalized state from Eq.~\ref{eq:state-prepare}, we can always prefix the data rows with the binary representation of indices $0$ to $n-1$, which will only inflate the length $m$ of each string by an additive factor of $\log_2 n$. From here onwards, we assume this to be the case. In the density matrix formulation, we write this map as 
\[
\rho_D  = \proj{\phi_D} 
\]
As shown in~\cite{schuld2021supervised}, the quantum encoding defines a \emph{quantum kernel} defined over two datasets $D$ and $D'$ as
\[
\kappa_\phi(D,D') := \mathsf{tr}(\rho_D\rho_{D'}) =  \braket{\phi_D}{\phi_{D'}} \mathsf{tr}( \ket{\phi_{D'}}\bra{\phi_D}) = |\braket{\phi_D}{\phi_{D'}}|^2,
\]
where $\mathsf{tr}()$ is the matrix trace function. 
For differential privacy,  the \emph{quantum minimum adjacent kernel} is defined in \cite{angrisani2022differential} as 
\[
\hat\kappa_\phi := \min_{D \sim D'} \kappa_\phi(D, D') 
\]
where recall that $D \sim D'$ denotes neighboring datasets. For the basis encoding, we get 
\begin{align*}
    \kappa_\phi(D, D') &= |\braket{\phi_D}{\phi_{D'}}|^2 \\
    &= \left| \frac{1}{n} \sum_i \sum_j \braket{x_i}{x_j} \right|^2 \\
    &= \left| \frac{1}{n} \sum_i \delta_{x_i, x_j} \right|^2 \\
    &= \left(  \frac{\sum_i \delta_{x_i, x_j}}{n}\right)^2
\end{align*}
From this it follows that 
$\hat\kappa_\phi$ for the basis encoding is
\[
\hat\kappa_\phi = \min_{D \sim D'} \kappa(D, D')  = \left(\frac{n-1}{n}\right)^2 = \left(1 - \frac{1}{n}\right)^2
\]
as shown in Table 1 in~\cite{angrisani2022differential}.\footnote{In~\cite{angrisani2022differential} this quantity is erroneously shown as $1 - 1/n$.} For any two states $\rho$ and $\sigma$, the trace ditance between them is defined as
\[
\mathcal{T}(\rho, \sigma) = \frac{1}{2}\mathsf{tr}(|\rho - \sigma|)
\]
where for a matrix $A$, we define $|A|:=\sqrt{A^\dagger A}$, i.e., the positive square root of $A^\dagger A$. 

\subsection{Differential Privacy Amplification}
\label{subsec:dp-amplify}
Several differential privacy amplification results are proved in~\cite{angrisani2022differential} related to the minimum adjacent kernel. We recall two relevant results here.

\begin{theorem}[\cite{angrisani2022differential}]
Let $D$ be a dataset and let $\rho_D  = \proj{\phi_D}$ be its quantum encoding. Let $\mathcal{M}$ be a mechanism that takes $\proj{\phi_D}$ as input, and not the original dataset $D$. Then $\mathcal{M} (\rho_D)$ is $(0, \sqrt{1 - \hat\kappa_\phi})$-differentially private.  
\end{theorem}
The above result is based on the observation that if $\rho_D$ and $\rho_{D'}$ are quantum encodings of two neighboring datasets $D$ and $D'$ then
\[
\mathcal{T}(\rho_D, \rho_{D'}) \leq \sqrt{1 - |\braket{\phi_D}{\phi_{D'}}|^2} = \sqrt{1 - \kappa(D, D')} \leq \sqrt{1 - \hat\kappa_\phi},
\]
where the last inequality follows from the definition of the minimum adjacent kernel. The quantity $\sqrt{1 - \hat\kappa_\phi}$ for basis encoding is
\[
\sqrt{1 - \left(1 - \frac{1}{n}\right)^2}= \frac{\sqrt{2n - 1}}{n}
\]

\begin{theorem}
\label{theorem:povm-lap}
Let $D$ be a dataset and let $\rho_D  = \proj{\phi_D}$ be its quantum encoding. Let $M$ be a positive operator-valued measure (POVM) with outcomes in $\{0, 1\}$. Let $\mathcal{M}$ be the following mechanism: (1) for $1 \leq i \leq t$, encode $D$ as $\rho_D$, (2) apply $\mathcal{M}$ to $\rho_D$ and store the outcome as $y_i$, (3) output $\frac{1}{t}\sum_{i=1}^t y_i + \eta$, where $\eta$ is a Laplace random variable of scale $\frac{1}{\epsilon}(\tau + \sqrt{1 - \hat\kappa_\phi} )$, for any $\tau \geq 0$. Then the mechanism $\mathcal{M}$ is $\epsilon$-differentially private with exponentially high probability in $\tau$ and $t$. 
\end{theorem} 

This result is based on the fact that if $\rho_D$ and $\rho_{D'}$ are quantum encodings of two neighboring datasets $D$ and $D'$ then 
\begin{equation}
\label{eq:mechanism-exp}
    |\mathbb{E}[\mathcal{M}(\rho_D) - \mathcal{M}(\rho_{D'})]| = |\Pr[\mathcal{M}(\rho_D) = 1] - \Pr[\mathcal{M}(\rho_{D'}) = 1] | \leq \sqrt{1 - \hat\kappa_\phi}
\end{equation}
Using the Chernoff-Hoeffding bound with parameter $\tau$ we can ensure that the expected values on the two datasets on the left hand side are exponentially close to their sample means, and hence the difference is at most $\tau + \sqrt{1 - \hat\kappa_\phi}$ with overwhelming probability. For basis encoding the Laplace mechanism uses the scale:
\begin{equation}
\label{eq:basis-encoding-lap-scale}
\frac{\tau + \sqrt{1 - \hat\kappa_\phi}}{\epsilon} = \frac{\tau + \sqrt{2n - 1}/n}{\epsilon}  
\end{equation}

\subsection{Preparing the Initial Encoded State}
Since the mechanisms will use quantum (basis) encoding of a dataset, we briefly discuss the complexity of preparing such a quantum state. In general, the dimension of the dataset, i.e., $m$, would be rather large, e.g., 100. However, the number of rows $n$ of the dataset usually satisfies $n \ll 2^m$. Hence, this corresponds to sparse initial state preparation, which can be done using $\mathcal{O}(mn)$ gates~\cite{de2022double-sparse, de2020circuit-pqm}. 
%~\cite{gleinig2021sparse}. The algorithm in~\cite{gleinig2021sparse} is much slower than this. The $\mathcal{O}(mn)$ construction result is in fact in~\cite{de2022double-sparse}. 
In particular the probabilistic quantum
memory (PQM) algorithm as described in~\cite{de2020circuit-pqm} exactly prepares the state from Eq.~\ref{eq:state-prepare}. 

\section{Predicate Queries and Quantum Circuits to Evaluate Them}
In this section, we recall the definition of predicate queries, and then discuss how they can be implemented as quantum circuits. Predicate queries arise naturally in a variety of real-world datasets, such as census data, and have been analyzed extensively in classical differential privacy literature~\cite{mckenna2021hdmm, smith2022making}. 

\subsection{Predicate Queries}
\label{subsec:pred-queries}
We assume that the dataset $D$ is vertically divided into \emph{attributes}. Examples of an attribute are profession and marital status. Formally, an \emph{attribute} $A$ is a finite set whose elements are called attribute \emph{values}. For instance, single and married are attribute values of the marital status attribute. Given the dataset $D = (x_1, x_2, \ldots, x_n)$, where $x_i \in \{0, 1\}^m$, we further assume that each attribute value of the dataset can be represented as a fixed $\ell$-bit binary string, so that we have $m/\ell$ attributes in the dataset. This can be done by padding the values of smaller attributes with extra zeros. 
A \emph{counting query} $q$ on a row is defined as a query whose range is $\{0, 1\}$, i.e., $q$ is a binary function of a row $x$. For the entire dataset, the counting query $q$ is defined as $q(D) \coloneqq \frac{1}{n}\sum_{i = 1}^n q(x_i)$. Our focus is on a restricted form of counting queries, called \emph{predicate queries}, defined as follows. A predicate on an attribute $A$ of the dataset is a boolean function. We extend this definition to the entire row $x \in D$, by assuming the predicate to evaluate to $1$ on all other attributes of the dataset. A predicate query is a conjunction and/or disjunction of predicates. Since a predicate query is a counting query, it follows that the global sensitivity of a predicate query $q$ is given as $\Delta q = 1/n$. 

\begin{example}
Consider a dataset with three attributes $A_1 = \{\texttt{Child}, \texttt{Adult}\}$, $A_2 = \{\texttt{Single}, \texttt{Married}, \texttt{Divorced}\}$ and $A_3 = \{\texttt{Teacher}, \texttt{Student}\}$. By assigning 2 bits per attribute, and encoding the attribute values of each attribute as a two-bit string, we can represent any possible row as a string from $\{0, 1\}^6$. For instance,
\[
\begin{matrix}
    \texttt{Child}: & 00\\
    \texttt{Adult}: & 01
\end{matrix} 
\;
\text{ and }
\;
\begin{matrix}
    \texttt{Single}: & 00\\
    \texttt{Married}: & 01 \\
    \texttt{Divorced}: & 10
\end{matrix} 
\;
\text{ and }
\;
\begin{matrix}
    \texttt{Teacher}: & 00\\
    \texttt{Student}: & 01
\end{matrix}\]
The following are examples of predicates:

\begin{itemize}
    \item $\phi_{\texttt{Single}}: x \texttt{ == **}00\texttt{**}$
    \item $\phi_{\texttt{Married}}: x \texttt{ == **}01\texttt{**}$
    \item $\phi_{\neg \texttt{Divorced}}: x \texttt{ != **}10\texttt{**}$ 
    \item $\phi_{\texttt{Adult}}: x \texttt{ == }01 \texttt{****}$
    \item $\phi_{\texttt{Teacher}}: x \texttt{ == ****}00$
\end{itemize}
where `$*$' denotes the \emph{don't care} character. The following are examples of predicate queries
\begin{itemize}
    \item $\phi_{\texttt{Single}} \wedge \phi_{\texttt{Teacher}}$
    \item $\phi_{\texttt{Adult}} \vee \phi_{\neg \texttt{Divorced}}$
    \item $(\phi_{\texttt{Adult}} \wedge \phi_{\texttt{Single}}) \vee \phi_{\texttt{Teacher}}$
\end{itemize}
\qed
\end{example}

\subsection{Quantum Circuits for Predicate Queries}
In practice, a predicate query will contain a constant number of predicates. Recall that a predicate is evaluated over a single attribute, which means that a predicate query will be evaluated over a constant number of attributes. In what follows we show circuits that combine predicates in a predicate query, and then circuits for the types of predicates considered in this paper.

\subsubsection{Combining Predicates}
Let $q = \phi_i \wedge \phi_j$ be a predicate query on attributes $i$ and $j$. The following reversible circuit computes this (AND) query.

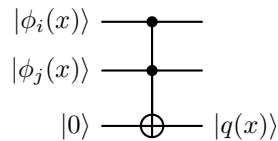
\begin{center}
\begin{quantikz}
\lstick{\ket{\phi_i(x)}} & \ctrl{2} & \\ 
\lstick{\ket{\phi_j(x)}} &  \ctrl{1} & \\
\lstick{\ket{0}} &  \targ{} & \rstick{\ket{q(x)}} \\
\end{quantikz}
\end{center}

where $x \in D$ is an arbitrary row of the dataset. On the other hand, if we have $q = \phi_i \vee \phi_j$ then the following reversible circuit computes this (OR) query.

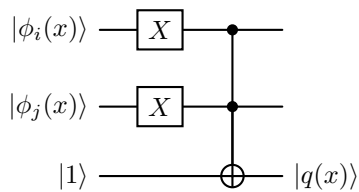
\begin{center}
\begin{quantikz}
\lstick{\ket{\phi_i(x)}} & \gate{X} & \ctrl{2} & \\ 
\lstick{\ket{\phi_j(x)}} &  \gate{X} & \ctrl{1} & \\
\lstick{\ket{1}} &  & \targ{} & \rstick{\ket{q(x)}} \\
\end{quantikz}
\end{center}

This follows from the fact that $\phi_i \vee \phi_j = \neg (\neg \phi_i \wedge \neg \phi_j)$. Thus a predicate query containing a combination of conjunctions and disjunctions can be evaluated using a combination of these two circuits. 

\subsubsection{Types of Predicates}
We are interested in two particular predicates, equality or inequality checks, and range predicates. These are standard predicates which are expected from census-like datasets~\cite{mckenna2021hdmm, smith2022making}. Both of these predicates can be described in terms of operations on a single bit (or qubit). Consider the equality predicate. That is, $\phi(x) = 1$ if $x = a$, and $0$ otherwise, where $a$ is some attribute value.\footnote{Since $x$ is a row, $x = a$ means that the target attribute of $x$ equals $a$, and the rest can take any value.} Let $a_i$ and $x_i$ denote the $i$th bit in the binary representation of $a$ and $x$, respectively. The following reversible circuit checks if $a_i = x_i$:

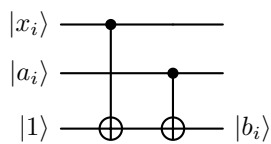
\begin{center}
\begin{quantikz}
\lstick{\ket{x_i}} & \ctrl{2} & & \\ 
\lstick{\ket{a_i}} & & \ctrl{1} & \\
\lstick{\ket{1}} & \targ{} & \targ{} & \rstick{\ket{b_i}} \\
\end{quantikz}
\end{center}

Now $x = a$ if $x_i = a_i$ for all $i \in [\ell]$ where $\ell$ is the size of the attribute. This can be achieved by using the circuit for AND query as follows:

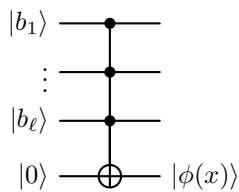
\begin{center}
\begin{quantikz}
\lstick{\ket{b_1}} & \ctrl{3} & \\ 
\lstick{\vdots} & \ctrl{2} & \\
\lstick{\ket{b_\ell}} &  \ctrl{1} & \\
\lstick{\ket{0}} &  \targ{} & \rstick{\ket{\phi(x)}} \\
\end{quantikz}
\end{center}

Similarly the predicate $x \neq a$ can be obtained using this circuit by using the ancilla bit $\ket{1}$ instead of $\ket{0}$. To check whether $a \leq x$, i.e., the predicate $\phi(x)$ which is $1$ if $a \leq x$ and $0$ otherwise, we can use the following circuit for each bit. This circuit first checks if $a_i \neq x_i$, and if it is true it flips the last qubit to $0$ if $a_i$ is $1$ but $x_i$ is $0$. 

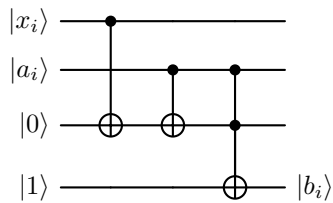
\begin{center}
\begin{quantikz}
\lstick{\ket{x_i}} & \ctrl{2} & & & \\ 
\lstick{\ket{a_i}} & & \ctrl{1} & \ctrl{2} &\\
\lstick{\ket{0}} & \targ{} & \targ{} & \ctrl{1} & \\
\lstick{\ket{1}} & & & \targ{}& \rstick{\ket{b_i}} \\
\end{quantikz}
\end{center}

We can then combine the results for each bit in $x$ to ascertain whether the predicate is true. The circuit can be tweaked in a straightforward manner to check whether $a \geq x$, and hence we can implement range predicates as well. 

\subsubsection{Cost of Quantum Circuits for Predicate Queries}
For each predicate query we have
\begin{itemize}
    \item A constant number of predicates 
    \item A multi-control Toffoli gate for combining the result of predicates
    \item Single ancillas for combining predicates
    \item A constant number of $X$ gates for joining predicates for an OR query. 
\end{itemize}
On a single predicate we have
\begin{itemize}
    \item A constant number of CNOT and Toffoli gates for each of the $\ell$ bits of the attribute
    \item A multi-control Toffoli gate to combine $\ell$ answers for each bit
    \item A constant number of ancillas per bit
\end{itemize}
Combining all this, we see that we can implement a predicate query using $\mathcal{O}(\ell)$ elementary gates using ancilla (work) qubits~\cite[\S 4.3]{nielsen-chuang}. Evaluating the query on the encoded dataset thus requires $\mathcal{O}(\ell n)$ time.

\section{Differentially Private Mechanisms for Quantum Encoded Datasets}
Recall that the predicate query is defined as the relative sum of the answers to individual rows in the dataset. Thus, given any $x \in \{0, 1\}^m$, we can define the query $q: \{0, 1\}^m \rightarrow \{0, 1\}$ as a Boolean function. Given the basis encoding of the dataset $D$
\[
\ket{\phi_D} = \frac{1}{\sqrt{n}}\sum_{i=1}^n \ket{x_i},
\]
where each $\ket{x_i}$ is a computational basis vector, we can write it as
\begin{equation}
\label{eq:query-phiD-decompose}
\ket{\phi_D} = \frac{1}{\sqrt{n}} \sum_{x \in D: q(x) = 1} \ket{x}  + \frac{1}{\sqrt{n}} \sum_{x \in D: q(x) = 0} \ket{x} 
\end{equation}
The vectors $\ket{x}$ are orthonormal basis vectors of the Hilbert space $\mathcal{H}^{\otimes m}$. The query function $q$ thus partitions $\mathcal{H}^{\otimes m}$ into two subspaces: one spanned by vectors $\ket{x}$ such that $q(x) = 1$, and the other spanned by vectors $\ket{x}$ satisfying $q(x) = 0$. We call them the good subspace and the bad subspace respectively following the language of~\cite{Brassard_2002}. Thus, we get a unique decomposition of $\ket{\phi_D}$ as 
\[
\ket{\phi_D} = \ket{\phi_\mathsf{G}} + \ket{\phi_\mathsf{B}}
\]
where $\ket{\phi_\mathsf{G}}$ belongs to the good subspace and $\ket{\phi_\mathsf{B}}$ to the bad subspace. Note further that
\[
\braket{\phi_\mathsf{G}}{\phi_\mathsf{G}} = \frac{1}{n} \left( \sum_{i \in [n] : q(x_i)  = 1} \bra{x_i} \right) \left(\sum_{j \in [n] : q(x_j)  = 1} \ket{x_j}\right) = \frac{1}{n} \sum_{i \in [n] : q(x_i)  = 1} \delta_{i, j} = \frac{q(D)}{n}  := \alpha
\]
And therefore 
\[
\braket{\phi_\mathsf{B}}{\phi_\mathsf{B}} = 1 - \alpha
\]
Following the discussion in the previous section, we can associate a unitary $U_q$ with a given query $q$ that works on $\ket{\phi_D}$ and an ancilla system $\ket{A}$ and registers the answer of $q(D)$ onto an ancilla bit of $A$. We can then apply a CNOT gate with the control as the qubit corresponding to $q(D)$ and the target as an additional ancilla bit, to obtain the final answer. Mathematically this transforms the system to
\[
(CX)  U_q \ket{\phi_D}\ket{0}\ket{A} \rightarrow (\ket{\phi_\mathsf{G}}\ket{1} + \ket{\phi_\mathsf{B}}\ket{0} ) \otimes \ket{A'}
\]
Next we describe two methods to 
measure the amplitude $\alpha$ of $\ket{\phi_\mathsf{G}}$ which gives the answer to the query. For simplicity, we ignore the ancilla system and assume that the system is in the state 
\begin{equation}
\label{eq:good-bad-state}
\ket{\phi_\mathsf{G}}\ket{1} + \ket{\phi_\mathsf{B}}\ket{0}
\end{equation}

\subsection{Privacy Amplification via Direct Measurement}
In the direct measurement method, we simply measure the ancilla bit in Eq.~\eqref{eq:good-bad-state} in the computational basis. The probability of observing the outcome 1 in a single measurement is $\alpha$. Thus,  over $t$ measurements we can get an estimate of $\alpha$ by averaging the outcomes of individual measurements. 
Furthermore, from Eq.~\eqref{eq:query-phiD-decompose} after each measurement, the first system in Eq.~\eqref{eq:good-bad-state} collapses to the states:
% $\ket{0}$ and apply an operator $Q$ (depending on the query) to obtain a state:
% \[
% Q\ket{\phi_D}\ket{0} = \ket{\phi_\mathsf{G}}\ket{1} + \ket{\phi_\mathsf{B}}\ket{0}
% \]
% From here, measuring the second system in the computational basis repeatedly gives us an approximation of $\alpha$. In the second method, we directly try to measure the amplitude of $\ket{\phi_\mathsf{G}}$ by applying the Grover type oracle and measuring the resulting eigenvalue.
\[
\frac{1}{\sqrt{\alpha n}} \sum_{x \in D: q(x) = 1} \ket{x}   \; \text{ or } \; \frac{1}{\sqrt{(1 - \alpha) n}} \sum_{x \in D: q(x) = 0} \ket{x}  
\]
Thus, in this method the client needs to send multiple encoded copies of the dataset to the server, one for each measurement. Let $\mathcal{M}$ denote this mechanism, and let $\rho_D$ and $\rho_{D'}$ denote the basis encoding of any two neighboring datasets. Let $\alpha = q(D)/n$ and $\alpha' = q(D')/n$. Then given that the query $q$ is a predicate query, we have
\[
| \Pr(\mathcal{M}(\rho_D) = 1) - \Pr(\mathcal{M}(\rho_{D'}) = 1)| = \left|\alpha - \alpha' \right| = \frac{1}{n}
\]
Putting this in Eq.~\eqref{eq:mechanism-exp}, and applying the Chernoff-Hoeffding inequality with parameter $\tau$, we see that with overwhelming probability (in $\tau$ and $t$), the output over any two neighboring datasets is within a factor:
\[
\frac{\tau + 1/n}{\epsilon}
\]
Comparing this to Eq.~\eqref{eq:basis-encoding-lap-scale}:
\[
\frac{\tau + \sqrt{1 - \hat\kappa_\phi}}{\epsilon} = \frac{\tau + \sqrt{2n - 1}/n}{\epsilon}  
\]
we see that the former is less than the latter for any $n > 1$. Thus, one way to answer predicate queries with differential privacy, is to follow the mechanism described in Theorem~\ref{theorem:povm-lap}, but with Laplace noise of scale $\frac{1}{\epsilon}(\tau + 1/n)$. 
Thus, we can apply Laplace noise of smaller scale when answering predicate queries than the scale given for a generic mechanism in~\cite{angrisani2022differential}.
However, both these scales are more than $1/\epsilon n$, which is the scale required to answer a predicate query using the standard differential privacy Laplace mechanism. Here, we show a more refined analysis which exploits the randomness induced by quantum measurements. 

\descr{Refined Privacy Analysis.} 
Our main observation is that measuring the second system in Eq.~\eqref{eq:good-bad-state} and submitting the result is the same as sampling a row $x$ uniformly at random from $D$ and returning the answer $q(x)$ as shown in Algorithm~\ref{algo:equiv-mech}. In both cases, the probability that the result is 1 is exactly $\alpha$. Thus, we can analyze the privacy loss through the following equivalent mechanism.

\begin{algorithm}
\SetAlgoLined
\DontPrintSemicolon
\SetKwInOut{Input}{Input}
\Input{Dataset $D$, query $q$, sampling parameter $t$, integer $k \geq 0$}
$\alpha \leftarrow 0$\;
\For{$i = 1$ \KwTo $t$}{
    $x \leftarrow_{\$} D$\;
    $\alpha \leftarrow \alpha + q(x)$\;
}
\Return $\frac{\alpha}{t} + \text{Lap}\left( \frac{k}{t\epsilon} \right)$
\caption{Equivalent Mechanism $\mathcal{M}$ to Measuring Eq.~\eqref{eq:good-bad-state} $t$ Times}
\label{algo:equiv-mech}
\end{algorithm}

In the following, we prove the following result.

\begin{theorem}
\label{theorem:equiv-mech-dp}
The mechanism $\mathcal{M}$ in Algorithm~\ref{algo:equiv-mech} is $(\epsilon'_k, \delta_k)$-differentially private, where $\epsilon'_k = 0$ if $k = 0$, and for all other values of $k$
\[
\epsilon'_k = \ln \left( \sum_{j=0}^k e^{j\epsilon/k} B(t, j)\right),
\]
and 
\[
\delta_k = 1 - \sum_{j=0}^k  B(t, j)
\]
where 
\[
B(t, j) = \binom{t}{j}\left(\frac{1}{n}\right)^i \left(1 - \frac{1}{n}\right)^{t-j}
\]
\qed
\end{theorem}
To prove this result, let us denote a generic pair of neighboring datasets by $D$ and $D'$. Let us further assume without loss of generality that the two differ in the $i$th row $x_i$. The answer to the query on the two datasets can differ by $1/n$ as shown in Section~\ref{subsec:pred-queries}.  Let $I$ be a random variable representing the number of times the row $x_i$ is sampled in the $t$ iterations of the mechanism in Algorithm~\ref{algo:equiv-mech}. Thus $I$ takes on values in the set $\{0, 1, \ldots, t\}$. Now take any subset $S$ of the range of $\mathcal{M}(D)$. We have
\begin{equation}
\label{eq:equiv-mech-prob}
\Pr[\mathcal{M}(D) \in S] = \sum_{j = 0}^t \Pr[\mathcal{M}(D) \in S \mid I = j] \Pr[I = j]    
\end{equation}
Clearly $I$ is binomially distributed as
\[
\Pr[I = j] = \binom{t}{j}\left(\frac{1}{n}\right)^i \left(1 - \frac{1}{n}\right)^{t-j} := B(t, j)
\]
The idea is that if we assume different values of $I$, we can prove a spectrum of $(\epsilon, \delta)$-DP guarantees from Eq.~\eqref{eq:equiv-mech-prob}. For instance, let us first see what happens if $I = 0$, i.e., $x_i$ is never sampled in the $t$ samples. We get
\begin{align}
    \Pr[\mathcal{M}(D) \in S] &= \sum_{j = 0}^t \Pr[\mathcal{M}(D) \in S \mid I = j] \Pr[I =j] \nonumber\\
    &= \Pr[\mathcal{M}(D) \in S \mid I = 0] \Pr[I = 0] + \sum_{j = 1}^t \Pr[\mathcal{M}(D) \in S \mid I = j] \Pr[I = j] \nonumber \\
    &= \Pr[\mathcal{M}(D') \in S] \Pr[I = 0] + \sum_{j = 1}^t \Pr[\mathcal{M}(D) \in S \mid I = j] \Pr[I = j] \nonumber \\
    &\leq \Pr[\mathcal{M}(D') \in S] \Pr[I = 0] + \sum_{j = 1}^t \Pr[I = j] \nonumber \\
    &= \Pr[\mathcal{M}(D') \in S] \Pr[I = 0] + (1 - \Pr[I = 0]) \nonumber\\
    &= \Pr[\mathcal{M}(D') \in S] (1 - \delta_0) + \delta_0 \nonumber\\
    &= \Pr[\mathcal{M}(D') \in S] - \Pr[\mathcal{M}(D') \in S] \delta_0 + \delta_0 \nonumber\\
    &\leq \Pr[\mathcal{M}(D') \in S] + \delta_0 \label{eq:zero-delta},
\end{align}
where we have used the fact that the answer to the query $q$ is the same on datasets $D$ and $D'$ if the $i$th row is never sampled and hence $\Pr[M(D) \in S \mid I = 0] = \Pr[M(D') \in S]$. We have further defined
\[
\delta_0 := 1 - \Pr[I = 0] = 1 - \left(\frac{n-1}{n}\right)^t = 1 - B(t, 0)
\]
By the symmetry of the argument, the same equation holds with the roles of $D$ and $D'$ interchanged. Hence, the mechanism is $(0, \delta_0)$-differentially private. In other words, we get differential privacy without adding explicit noise! This value of $\delta_)$ however, is rather high. For example, if $n = 10^6$ and $t = 10^3$, then we get $\delta_0 = 0.0009995 < 10^{-3}$. Thus, while this gives us DP, we cannot use such a low value of $\delta_0$ for repeated application of this mechanism, since then the $\delta$'s add up through the composition theorem (Theorem~\ref{prop:composition}).

Returning to Eq.~\eqref{eq:equiv-mech-prob}, let us now assume that only $I = 0$ and $I = 1$ are likely. We get

\begin{align}
    \Pr[\mathcal{M}(D) \in S] &= \sum_{j = 0}^t \Pr[\mathcal{M}(D) \in S \mid I = j] \Pr[I = j] \nonumber\\
    &= \Pr[\mathcal{M}(D) \in S \mid I = 0] \Pr[I = 0] + \Pr[\mathcal{M}(D) \in S \mid I = 1] \Pr[I = 1] \nonumber \\
    &+ \sum_{j = 2}^t \Pr[\mathcal{M}(D) \in S \mid I = j] \Pr[I = j] \nonumber \\
    &= \Pr[\mathcal{M}(D') \in S] \Pr[I = 0] + \Pr[\mathcal{M}(D) \in S \mid I = 1] \Pr[I = 1] \nonumber \\
    &+ \sum_{j = 2}^t \Pr[M(D) \in S \mid I = j] \Pr[I = j] \nonumber
\end{align}
where we have once again used the fact that if the event $I = 0$ occurs, then the outcome of $\mathcal{M}$ is the same over the two datasets. Now if the event $I = 1$ occurs, $x_i$ occurs exactly once in the $t$ iterations. Thus, the answer to $q$ can change by at most $1/t$ in the neighboring datasets. Thus, adding Laplace noise of scale $1/t\epsilon$, makes this mechanism $(\epsilon, 0)$-DP. Therefore
\[
\Pr[\mathcal{M}(D) \in S \mid I = 1] \leq e^\epsilon \Pr[\mathcal{M}(D') \in S]
\]
and the above equation becomes
\begin{align}
    \Pr[\mathcal{M}(D) \in S] &\leq \Pr[\mathcal{M}(D') \in S] \Pr[I = 0] + e^\epsilon  \Pr[\mathcal{M}(D') \in S]\Pr[I = 1] \nonumber \\
    &+ \sum_{j = 2}^t \Pr[\mathcal{M}(D) \in S \mid I = j]\Pr[I = j] \nonumber\\
    &=  (\Pr[I = 0] + e^\epsilon  \Pr[I = 1])\Pr[M(D') \in S] \nonumber \\
    &+ \sum_{j = 2}^t \Pr[\mathcal{M}(D) \in S \mid I = j] \Pr[I = j]\nonumber \\
    &= \left( \left(\frac{n-1}{n}\right)^t + e^\epsilon \left(\frac{t}{n}\left(\frac{n-1}{n}\right)^{t-1}\right)\right) \Pr[\mathcal{M}(D') \in S] \nonumber\\
    &+ \sum_{j = 2}^t \Pr[\mathcal{M}(D) \in S \mid I = j] \Pr[I = j]\nonumber \\
    &= \left(\frac{n-1}{n}\right)^{t-1} \left(1 - \frac{1}{n} + \frac{te^\epsilon}{n} \right) \Pr[\mathcal{M}(D') \in S] \nonumber \\
    &+ \sum_{j = 2}^t \Pr[\mathcal{M}(D) \in S \mid I = j] \Pr[I = j]\nonumber\\
    &= e^{\epsilon'_1} \Pr[\mathcal{M}(D') \in S] + \sum_{j = 2}^t \Pr[\mathcal{M}(D) \in S \mid I = j] \Pr[I = j]\nonumber\\
    &\leq e^{\epsilon'_1} \Pr[\mathcal{M}(D') \in S] + \sum_{j = 2}^t \Pr[I = j]\nonumber\\
    &\leq e^{\epsilon'_1} \Pr[\mathcal{M}(D') \in S] + \delta_1 \label{eq:two-terms}
\end{align}
where
\begin{align}
 \delta_1 &:= \sum_{j = 2}^t \Pr[I = j] \nonumber\\
        &= 1 - \Pr[I = 0] - \Pr[I = 1]  \nonumber\\
        &= 1 - \sum_{j = 0}^1 B(t, j) \nonumber\\
        &= 1 - \left(\frac{n-1}{n}\right)^t  -  \frac{t}{n}\left(\frac{n-1}{n}\right)^{t-1}  \nonumber\\
        &= 1 - \left(\frac{n-1}{n}\right)^{t-1}\left( \frac{n + t - 1 }{n}\right) \label{eq:delta-sampling}
\end{align}
and 
\begin{align}
    e^{\epsilon'_1} = \left(\frac{n-1}{n}\right)^{t-1} \left(1 - \frac{1}{n} + \frac{te^\epsilon}{n} \right) \nonumber\\
\Rightarrow \epsilon'_1 = (t-1)\ln\left( \frac{n-1}{n} \right) + \ln \left(1 - \frac{1}{n} + \frac{te^\epsilon}{n} \right) \label{eq:eps-sampling}
\end{align}
For concrete values, if $t = 10^3$ and $n = 10^6$, then with $\epsilon = 1$, we get 
\[
\epsilon'_1 \approx 0.0017146 \text{ and }\delta_1 \approx 0.5 \times 10^{-6}
\]
Note that this corresponds to $k = 1$ in Algorithm~\ref{algo:equiv-mech}. In general, for any $k \geq 1$, we add noise $\text{Lap}(k/t\epsilon)$ to the answer in the last step of the algorithm. This is justified since for any event $I = j$, with $j \leq k$, the sample $x_i$ occurs exactly $j$ times in the $t$ iterations. This means that the answer to $q$ can change by at most $j/t$ in the two datasets. Adding Laplace noise of scale $k/t\epsilon$, implies Laplace noise of scale 
\[
\frac{j}{t\epsilon(\frac{j}{k})},
\]
and hence the mechanism is $(\frac{j}{k}\epsilon, 0)$-DP. With this refinement we get that for any $k \geq 0$, we get 
\[
    \epsilon'_k = \begin{cases}
        0, & \text{ if } k = 0 \\
        \ln \left( \sum_{j=0}^k e^{j\epsilon/k} B(t, j)\right) & \text{ if } k \geq 1
    \end{cases}
\]
% \[
% \epsilon'_k = \ln \left( \sum_{j=0}^k e^{j\epsilon/k} B(t, j)\right),
% \]
and 
\[
\delta_k = 1 - \sum_{j=0}^k  B(t, j)
\]
where $B(t, j)$ is the $j$th term of the binomial distribution related to $I$, i.e., $\Pr[I = j]$. For example, with $k = 2$ and $\epsilon = 1$, we get 
\[
\epsilon'_2 \approx 0.0006487, \delta'_2 \approx 1.6604 \times 10^{-10},
\]
with the tradeoff that we need to add noise of higher scale ($k = 2$ times more) in the algorithm. 

To recap, we have shown that the mechanism of direct measurement, i.e., one that measures the state in Eq.~\eqref{eq:good-bad-state} a total of $t$ times, and returns the average of outcomes with Laplace noise of scale $k/t\epsilon$ is
\begin{itemize}
    \item $(0, \delta_0)$-DP with $k = 0$. In other words, the mechanism satisfies differential privacy without adding explicit noise. On the downside it cannot be used for multiple queries as practical values of $\delta_0$ are rather high.
    \item $(\epsilon'_k, \delta_k)$-DP with $k \geq 1$, where both $\epsilon'_k$ and $\delta_k$ are described in the statement of Theorem~\ref{theorem:equiv-mech-dp}. The tradeoff is that we add noise of higher scale with larger $k$, but the values of $\epsilon'_k$ and $\delta_k$ can be made small. Crucially, this refined analysis shows that much less noise needs to be added through this mechanism than the simpler analysis at the start of this section or through Eq.~\ref{eq:basis-encoding-lap-scale} derived from the work in~\cite{angrisani2022differential}. For instance, if $t = 10^3$ is enough to give us a good estimate of the true answer, then adding noise of scale $\text{Lap}(2/t\epsilon)$, with $\epsilon = 1$ (and $k = 2$), gives us a privacy budget consumption of just $\epsilon'_2 \approx 0.0006487$ as stated above. 
\end{itemize}

\subsection{Amplitude Estimation via Phase Estimation}
\label{subsec:amp-est}
Another way to measure the query answer $\alpha$ is via amplitude estimation methods~\cite{suzuki2020amplitude, uno2021modified, callison2023improvedmaximumlikelihoodquantumamplitude, aaronson2020quantum, giurgica2022low, grinko2021iterative}. This follows because in Eq.~\eqref{eq:good-bad-state}, $\alpha$ is the amplitude of the state $\ket{\phi_\textsf{G}}$.\footnote{Note that we are deliberately cavalier in calling $\alpha$ both the amplitude and its square as the distinction does not matter in our case.} Among the amplitude estimation methods, the canonical method by Brassard et al \cite{Brassard_2002} suffices for our purpose assuming an error corrected and well controlled quantum computing environment. Another reason to choose this method is because it requires only 24 initial states or independent sampling to get a confidence level more than 99\% in estimating the amplitude. Before we discuss how the amplitude estimation method can be made differentially private, we discuss its efficiency as compared to the direct measurement method in the previous section. 

%Suppose we outsource the whole calculation to the server, ie, we are going to send limited number of copies of the initial state and the server will make the measurement and estimate the amplitude by the amplitude estimation methods \cite{brakerski2018quantum,suzuki2020amplitude, uno2021modified, maronese2023quantum, callison2023improvedmaximumlikelihoodquantumamplitude, aaronson2020quantum, giurgica2022low, grinko2021iterative}. Why amplitude estimation?! It is because the we want to see the query result that is $q(X)= \alpha$ which is precisely the amplitude of the state after running our query algorithm that we can see in \eqref{eq:query-phiD-decompose}. Among these methods the original canonical amplitude estimation method by Brassard et al.\cite{Brassard_2002} is most efficient for our purpose or in general if it happens to be a most ideal error corrected and well controlled quantum computing environment. Before we go into how this amplitude estimation method can help us with differential privacy, we must discuss how it will be efficient in terms of efficient computing and how does it work. 

Given the decomposed state
\[
\ket{\phi_D} = \ket{\phi_\textsf{G}} + \ket{\phi_\textsf{B}}
\]
from Eq.~\eqref{eq:good-bad-state} (ignoring the ancilla system), amplitude estimation applies an operator defined as $Q = -AS_0A^{\dagger}S_\mathsf{G}$, where $S_0 = \mathbb{I}_n - 2\ket{0}_n\bra{0}_n$ and $S_\textsf{G} = \mathbb{I}_n - 2\proj{\phi_\textsf{G}}$. Furthermore, $A$ is the quantum Fourier transform defined as
\[
\text{QFT}_M: |x\rangle
\rightarrow  \frac{1}{\sqrt{M}} \sum_{y=0}^{M-1} e^{i 2\pi  x y / M}\, |y\rangle
\]
% The canonical amplitude estimation method can be stated as for an operator $A$ acting on $n$ qubit register as:
% \[
% A\ket{0 } = \sqrt{1-\alpha}\ket{\Psi_B}+\sqrt{\alpha}\ket{\Psi_G}. 
% \]
% This $A$ can be thought of as the initial state followed by the \texttt{CNOT} operation in our case (in our case the register will consist $n+1$ qubits so you can see the 'good' states ie, \ket{\Psi_G} are all the odd states but anyway it will not change anything in the calculation so we are considering the register will have $n$ qubit). Let us define an operator $Q = -AS_0A^{\dagger}S_G$ such that $S_0 = \mathbb{I}_n - 2\ket{0}_n\bra{0}_n$ and $S_G = \mathbb{I}_n - 2\ket{\Psi_G}_n\bra{\Psi_G}_n$. 
The operator $Q$ works as a Grover type query operator which acts $M$ times to estimate the amplitude. This method requires $\mathcal{O}(1/\Delta)$ query calls (the parameter $M$) to estimate the amplitude within a $\Delta$ error bound, whereas the direct measurement method requires $\mathcal{O}(1/\Delta^2)$ measurements for the same error bound (see Appendix~\ref{app:comp-analysis}). So, this is a quadratic speedup. To be more precise, if the real amplitude is $\alpha$ and the estimated amplitude is $\widetilde{\alpha}$ then~\cite[Theorem 12]{Brassard_2002},
\[
|\alpha - \widetilde{\alpha}| \leq \frac{2\pi\sqrt{\alpha(1-\alpha)}}{M}+\frac{\pi^2}{M^2} = \Delta = \mathcal{O}\left(\frac{1}{M}\right)
\]
with probability (confidence) $\geq \frac{8}{\pi^2}$. The usual circuit to implement this algorithm is given in Figure~\ref{fig:can-amp-est}.

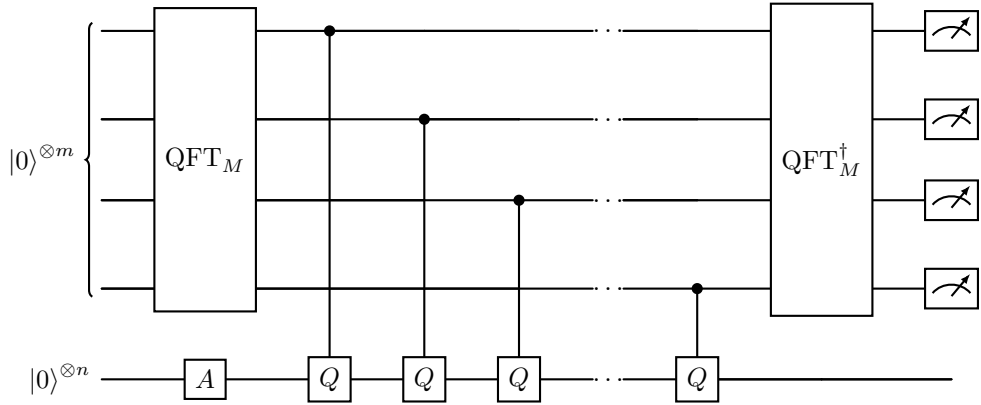
\begin{figure}
    \centering
    \begin{quantikz}[row sep=0.55cm, column sep=0.7cm]
\lstick[wires=4]{$\ket{0}^{\otimes m}$}
  & \gate[wires=4]{\text{QFT}_M}
  & \ctrl{4} & \qw         & \qw     & \push{\hdots}    & \qw              & \gate[wires=4]{\text{QFT}^{\dagger}_M} & \meter{} \\
  & \qw
  & \qw     & \ctrl{3}     & \qw      & \push{\hdots}         & \qw              &                              & \meter{} \\
  & \qw
  & \qw     & \qw          & \ctrl{2}    & \push{\hdots}    & \qw              &                              & \meter{} \\
  & \qw
  & \qw     & \qw          & \qw         & \push{\hdots}     & \ctrl{1}         &                              & \meter{} \\
\lstick{$\ket{0}^{\otimes n}$}
  & \gate{A}
  & \gate{Q} & \gate{Q} & \gate{Q}  & \push{\hdots}  & \gate{Q} & \qw & \qw
\end{quantikz}
\caption{Circuit to implement canonical quantum amplitude estimation}
\label{fig:can-amp-est}
\end{figure}

\descr{Improving the Confidence Bound.} The confidence level of $8/\pi^2$ can be made arbitrarily close to 100\% by using the amplitude estimation algorithm independently over multiple initial states, and then taking the median value. To see this let $X_i$ denote a binary random variable which equals $1$ if $|\alpha - \widetilde{\alpha}| > \Delta$ in the $i$th run of the amplitude estimation algorithm. As shown above, this happens with probability $1 - p \leq 1 - 8/\pi^2$. Let $X = \sum_{i = 1}^t X_i$. Through median estimation, we take the median of the $t$ estimations as the value for $\alpha$. The median estimation will fail, i.e., this value will have additive error $> \Delta$, if $X \geq t/2$. Thus we are interested in knowing $\Pr[X \geq t/2]$. This probability is maximized if each $X_i$ is $1$ with probability $1 - p$ exactly equal to $1 - 8/\pi^2$. Assuming this to be the case, $X$ is binomially distributed with mean $(1-p)t$. By Hoeffding's inequality
\begin{align*}
    \Pr\left[X \geq \frac{t}{2}\right] &= 
    \Pr\left[X - (1-p)t \geq \frac{t}{2} - (1 -p)t \right]\\
    &= \Pr\left[X - (1-p)t\geq t\left( p - \frac{1}{2}\right)\right]\\
    &\leq \exp\left( - \frac{2(t(p-1/2))^2}{t}\right) \\
    &= \exp(-2t(p-1/2)^2)\\
    &= \exp\left( -2 t \left( \frac{8}{\pi^2} - \frac{1}{2}\right)^2\right)
\end{align*}
Thus, for instance, to improve the probability of success from $8/\pi^2 \approx 0.81$ to $0.99$, we can use $t = 24$ repetitions.

\descr{Details of the Algorithm.} To understand where we can apply differential privacy, let us dive deeper into the algorithm. 
%If we run canonical AE algorithm on a database to find out the amplitude of good state of a desired query, we need to understand what answer we may expect after the measurement and to determine so we must see the state before the measurement so that we can expect what type of outcomes are possible as an outcome. Let us assume the initial state is $\ket{\Psi}=A\ket{0^n}$ which is the data set after getting encoded into the quantum state. 
As in \cite{Brassard_2002}, and detailed at the start of this section, $\ket{\phi_D}$ can be broken down into the good and the bad states as 
\[
\ket{\phi_D} = \ket{\phi_\mathsf{G}} + \ket{\phi_\mathsf{B}}
\]
Define 
\[
\alpha := \braket{\phi_\mathsf{G}}{\phi_\mathsf{G}} = \sin^2 \theta_\alpha, \quad 0 \leq \theta_\alpha \leq \frac{\pi}{2}
\]
% \[
% \ket{\Psi}=\ket{\Psi_G}+\ket{\Psi_B},\qquad
% a:=\langle \Psi_G|\Psi_G \rangle =  \sin^2\theta_{\alpha}, \quad 0\le \theta_{\alpha}\le \frac{\pi}{2}.
% \]
The vector $\ket{\phi_D}$ lies in the subspace $H_{\phi_D}=\mathrm{span}\{\ket{\phi_\mathsf{G}},\ket{\phi_\mathsf{B}}\}$. The operator $Q$ leaves the subspace $H_{\phi_D}$ invariant~\cite[Lemma 1]{Brassard_2002}. Since $Q$ is unitary, $H_{\phi_D}$ has an orthonormal basis consisting of the two eigenvectors of $Q$ given as:
\[
\ket{\phi_\pm}
=\frac{1}{\sqrt{2}}
\left(\frac{\ket{\phi_\mathsf{G}}}{\sqrt{{\alpha}}} \pm i\,\frac{\ket{\phi_\mathsf{B}}}{\sqrt{1-{\alpha}}}\right),
\]
with eigenvalues
\[
Q\ket{\phi_\pm}=e^{\pm i2\theta_{\alpha}}\ket{\phi_\pm}
\]
% With $S_G$ flip the phase of good states, and $S_0$ flip only $\ket{0^n}$, let us define the Grover type oracle
% \[
% Q \;=\; -A\,S_0\,A^{\dagger}\,S_G.
% \]
% And the eigenvectors of the Hilbert space, $H_\Psi=\mathrm{span}\{\ket{\Psi_1},\ket{\Psi_0}\}$,
% \[
% \ket{\Psi_\pm}
% =\frac{1}{\sqrt{2}}
% \left(\frac{\ket{\Psi_G}}{\sqrt{{\alpha}}} \pm i\,\frac{\ket{\Psi_B}}{\sqrt{1-{\alpha}}}\right),
% \qquad
% Q\ket{\Psi_\pm}=e^{\pm i\,2\theta_{\alpha}}\ket{\Psi_\pm}
% \]
Under the $\ket{\phi_\pm}$ basis, $\ket{\phi_D}$ can be expressed as 
\[
\ket{\phi_D}=
%A\ket{0}^{\otimes n}
-\frac{i}{\sqrt{2}}\Big(e^{i\theta_{\alpha}}\ket{\phi_+}-e^{-i\theta_{\alpha}}\ket{\phi_-}\Big).
\]
Brassard et al.~\cite{Brassard_2002} define the operator $\Lambda_M$ which acts on any unitary operator $U$ as 
\[
\ket{j}\ket{y} \rightarrow \ket{j}U^j\ket{y}, \quad 0 \leq j < M
\]
With this definition, the operator $\Lambda_M(Q)$ on the state $\ket{j}\ket{\phi_\pm}$, where $\ket{\phi_\pm}$ are the eigenvectors of $Q$, results in $\ket{j}Q^j\ket{\phi_\pm} = \ket{j} e^{\pm i 2j\theta_\alpha}\ket{\phi_\pm}$. They further define
\[
S_M(\omega) = \frac{1}{\sqrt{M}} \sum_{y=0}^{M-1} e^{i 2\pi  \omega y}\, |y\rangle
\]
for $0 \leq \omega < 1$, from which it follows that 
\[
\text{QFT}_M \ket{x} = S_M(x/M)
\]
Note that by definition 
\[
\text{QFT}^\dagger_M (S_M(x/M)) = \ket{x}
\]
We can now describe the working of their amplitude estimation algorithm. We first prepare the whole state including the ancilla system as 
\[
\ket{\phi_D^{(1)}}
=\ket{0}^{\otimes m}\otimes\left[-\frac{i}{\sqrt{2}}\Big(e^{i\theta_{\alpha}}\ket{\phi_+)}-e^{-i\theta_{\alpha}}\ket{\phi_-}\Big)\right]
\]
where the second system is $\phi_D$ in the basis $\ket{\phi_\pm}$ as shown above. We then apply $\text{QFT}_M$ to the first register to obtain:
\[
\ket{\phi_D^{(2)}}
=\frac{1}{\sqrt{2M}}\sum_{j=0}^{M-1}\ket{j}\otimes
\Big(e^{i\theta_{\alpha}}\ket{\phi_+}-e^{-i\theta_{\alpha}}\ket{\phi_-}\Big).
\]
Next we apply $\Lambda_M(Q)$ on this state, and use the fact that $Q^j\ket{\phi_\pm} = e^{\pm i 2j\theta_\alpha}\ket{\phi_\pm}$
 to get
\begin{align}
\ket{\phi_D^{(3)}}
&=\frac{1}{\sqrt{2M}}\sum_{j=0}^{M-1}\ket{j}\otimes
\Big(e^{i\theta_{\alpha}}e^{+i 2j\theta_{\alpha}}\ket{\phi_+}-e^{-i\theta_{\alpha}}e^{-i2j\theta_{\alpha}}\ket{\phi_-}\Big)\nonumber\\
&=\frac{e^{i\theta_{\alpha}}}{\sqrt{2}}\left(\frac{1}{\sqrt{M}}\sum_{j=0}^{M-1}e^{i2 j \theta_{\alpha}}\ket{j}\right)\!\otimes\ket{\phi_+}
\;-\;
\frac{e^{-i\theta_{\alpha}}}{\sqrt{2}}\left(\frac{1}{\sqrt{M}}\sum_{j=0}^{M-1}e^{-i 2 j \theta_{\alpha}}\ket{j}\right)\!\otimes\ket{\phi_-}\nonumber\\
&=\frac{e^{i\theta_{\alpha}}}{\sqrt{2}}\;\ket{S_M(\tfrac{\theta_{\alpha}}{\pi})}\otimes\ket{\phi_+}
\;-\;
\frac{e^{-i\theta_{\alpha}}}{\sqrt{2}}\;\ket{S_M(1-\tfrac{\theta_{\alpha}}{\pi})}\otimes\ket{\phi_-} \label{eq:phiD_3}
\end{align}
which follows since  $e^{-i2j\theta_{\alpha}}=e^{i2\pi (-\theta_{\alpha}/\pi)j}=e^{i2\pi (1-\theta_{\alpha}/\pi)j}$ for integer $j$.
At this point we apply $\text{QFT}_M^{\dagger}$ to the first register and measure the first system in the computational basis. As shown in~\cite{Brassard_2002}, this gives us a value of $y$ in the range $0 \leq y < M$, and using the relation $\frac{\theta_\alpha}{\pi} \approx \frac{y}{M}$, we get the estimate of $\alpha$ as $\widetilde{\alpha} = \sin^2(\pi y/M)$. 

To make this mechanism differentially private, we need to find the maximum change in $\theta_\alpha$ given that the global sensitivity of $q(D)/n = \alpha$ is $1/n$. We can then apply Laplace noise of appropriate scale.

\subsubsection{Global Sensitivity Analysis}
As before let $\alpha = q(D)/n$ and $\alpha' = q(D')/n$, where $D$ and $D'$ are neighboring  datasets. Define $\theta_\alpha$ and $\theta_{\alpha'}$ as $\alpha = \sin^2 \theta_\alpha$ and $\alpha' = \sin^2 \theta_{\alpha'}$, respectively where $0 \leq \theta_\alpha, \theta_{\alpha'} \leq \pi/2$. We already know that
\[
| \alpha - \alpha'| = \frac{1}{n}
\]
We are interested in finding an upper bound on 
\[
|\theta_\alpha - \theta_{\alpha'}| \quad \text{ given } | \alpha - \alpha'| = |\sin^2 \theta_\alpha - \sin^2 \theta_{\alpha'}| = \frac{1}{n}
\]
Let us define $f(x) = \sin^2 x$, where $0 \leq x \leq \pi/2$. This function is differentiable with the derivative:
\[
f'(x) = 2\sin x \cos x = \sin 2 x
\]

Note that $\sin^2 x$ is a concave function on the interval $[\pi/4, \pi/2]$, since its derivative, i.e., $\sin 2x$, is decreasing on this interval~\cite{thomas2014thomas} (see Figure~\ref{fig:sin2x-max}). 
%A continuous function $f$ is concave on the interval $I$ if for any two points $x_1, x_2 \in I$, and any $\lambda \in [0, 1]$ we have
%\[
%f(\lambda x_1 + (1 - \lambda)x_2) \geq \lambda f(x_1) + (1-\lambda)f(x_2)
%\]
%See for example [WOLFRAM CONVEX]. 

\begin{figure}
    \centering
\begin{tikzpicture}
    
  \begin{axis}[
    domain=0:pi/2,
    samples=200,
    xmin=0, xmax=1.57079632679,   % pi/2
    ymin=0, ymax=1,
    axis lines=middle,
    xlabel={$\theta$},
    ylabel={$\sin^2 \theta$},
    xtick={0, pi/4, pi/2},
    xticklabels={$0$, $\tfrac{\pi}{4}$, $\tfrac{\pi}{2}$},
    ytick={0,1/2,1},
    yticklabels={0, 1/2, 1},
    grid=both,
    grid style={dashed,gray!30},
    width=8cm,
    height=6cm,
    enlargelimits=false,
    clip=true,
    axis line style={thick, ->, >={Stealth[, length=3mm, width=2mm]}},
    axis on top=true,
  ]

    % Dot coordinates
    \pgfmathsetmacro{\xdot}{pi/4 + 0.35}
    \pgfmathsetmacro{\ydot}{0.822}

    % Shaded regions behind the curve
    \fill[gray!30,opacity=0.7] (axis cs:\xdot,0) rectangle (axis cs:pi/2,\ydot);  % bottom-right
    \fill[gray!30,opacity=0.7] (axis cs:\xdot,\ydot) rectangle (axis cs:pi/2,1);  % top-right
    \fill[gray!30,opacity=0.7] (axis cs:0,\ydot) rectangle (axis cs:\xdot,1);    % top-left
    % Bottom-left rectangle is optional; leave unshaded to keep L-shape

    % Plot sin^2 x on top
    \addplot [thick,domain=0:pi/2] {sin(deg(x))^2};

    % Dot on top of everything
    \addplot+[only marks, mark=*, mark options={fill=black,draw=black}] coordinates { (\xdot, \ydot) };
    % Draw a horizontal double-headed arrow with n on top
    \draw[<->, thick] (113,48) -- (157,48) node[midway, above] {$\sin^{-1}\left( \frac{1}{\sqrt{n}} \right)$};
    \draw[<->, thick] (81,82) -- (81,100) node[midway, left] {$\frac{1}{n}$};
  \end{axis}
\end{tikzpicture}
    \caption{An illustration of the fact that if $|\sin^2 \theta_1 - \sin^2 \theta_2| = 1/n$, for any two angles $\theta_1, \theta_2 \in [0, \pi/2]$, then the maximum possible absolute distance between the angles is achieved in the period where the slope of $\sin^2 \theta$ is the flattest. This corresponds with one of the end-points of the period being either $\pi/2$ or $0$ due to the symmetry of the function around $\pi/4$.}
    \label{fig:sin2x-max}
\end{figure}
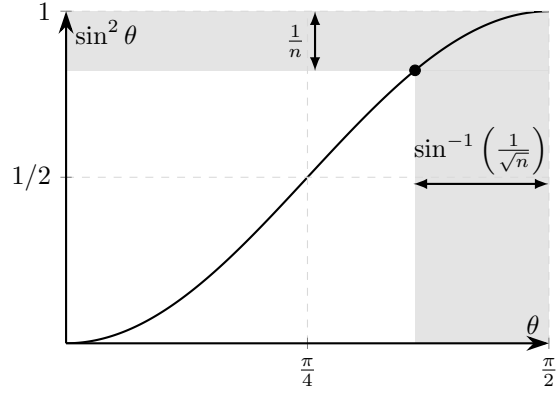

The secant slope is the average rate of change of a function $f$ on an interval $[x_1, x_2]$:
\[
\frac{f(x_2) - f(x_1)}{x_2 - x_1}
\]

\begin{theorem}[Mean Value Theorem~\cite{thomas2014thomas}]
\label{theorem:mvt}
Suppose $f(x)$ is continuous over an interval $[x_1, x_2]$ and differentiable on $(x_1, x_2)$ then there is at least one point $c \in (x_1, x_2)$ such that 
\[
f'(c) = \frac{f(x_2) - f(x_1)}{x_2 - x_1}
\]
\qed
\end{theorem}
Let $\theta_\alpha = \pi/2$ and $\theta_{\alpha'} = \pi/2 - \beta$, where $0 < \beta \leq \pi/2$ is such that $|\sin^2 \theta_\alpha - \sin^2 \theta_{\alpha'}| = \frac{1}{n}$. This implies that 
\[
| \theta_\alpha - \theta_{\alpha'}| = \beta,
\]
and 
\begin{align*}
    |\sin^2 \theta_\alpha - \sin^2 \theta_{\alpha'}| &= 
    \left|\sin^2 \left( \frac{\pi}{2}\right) - \sin^2 \left( \frac{\pi}{2} - \beta\right)\right| \\
    & = \left| 1 - \cos^2 \beta\right| \\
    &= \sin^2 \beta
\end{align*}
This gives us 
\[
\sin^2\beta = \frac{1}{n} \Rightarrow \beta = \sin^{-1}\left( \frac{1}{\sqrt{n}} \right)
\]
In the following we show that any other interval $| \theta_\alpha - \theta_{\alpha'}|$ satisfying $|\sin^2 \theta_\alpha - \sin^2 \theta_{\alpha'}| = \frac{1}{n}$, implies $| \theta_\alpha - \theta_{\alpha'}| \leq \beta$. Hence this is the sensitivity of the angle over all neighboring datasets. To prove this, we start with a few lemmas.

\begin{lemma}
\label{lemma:ave-slope}
Let $f$ be a concave function on an interval $[x_1, x_2]$ and differentiable on $(x_1, x_2)$. Let $a, b, c, d \in [x_1, x_2]$ such that $a < b < c < d$. Then,
\[
\frac{f(c) - f(a)}{c - a} \geq \frac{f(d) - f(c)}{d - c}
\]
\end{lemma}
\begin{proof}
Consider intervals $[a, b], [b, c]$ and $[c, d]$. By the mean value theorem (Theorem~\ref{theorem:mvt}), There exist $s_{a,b}$, $s_{b,c}$ and $s_{c,d}$ such that 
\[
s_{a,b} = \frac{f(b) - f(a)}{b - a}, s_{b,c} = \frac{f(c) - f(b)}{c - b}, s_{c,d} = \frac{f(d) - f(c)}{d - c},
\]
Since the function is concave on $[x_1, x_2]$, its derivative $f'$ is decreasing on this interval~\cite{thomas2014thomas}. Therefore since $s_{a,b}$, $s_{b,c}$ and $s_{c,d}$ are evaluated on points in $(a, b)$, $(b, c)$ and $(c, d)$ respectively, it follows that
\[
s_{a,b} \geq s_{b,c} \geq s_{c,d}
\]
Now observe that 
\begin{align*}
   \frac{f(c) - f(a)}{c - a} &= \frac{c-b}{c-a} \frac{f(c) - f(b)}{c - b} + \frac{b-a}{c-a} \frac{f(b) - f(a)}{b - a} \\
   &= \frac{c-b}{c-a} s_{b,c} + \frac{b-a}{c-a} s_{a,b} \\
   &\geq \frac{c-b}{c-a} s_{b,c}  +  \frac{b-a}{c-a}s_{b,c} \\
   &= s_{b,c}
\end{align*}
Similarly,
\begin{align*}
\frac{f(d) - f(b)}{d - b} &= \frac{d-c}{d-b} \frac{f(d) - f(c)}{d - c} + \frac{c-b}{d-b} \frac{f(c) - f(b)}{c - b} \\
&= \frac{d-c}{d-b} s_{c,d} + \frac{c-b}{d-b} s_{b,c} \\
&\leq \frac{d-c}{d-b} s_{b,c} + \frac{c-b}{d-b} s_{b,c} \\
&=  s_{b,c}
\end{align*}
Combining the two we get 
\[
\frac{f(c) - f(a)}{c - a} \geq \frac{f(d) - f(b)}{d - b},
\]
as desired.
\end{proof}

\begin{lemma}
\label{lemma:ave-slope-2}
Let $f$ be a concave function on an interval $[x_1, x_2]$ and differentiable on $(x_1, x_2)$. Let $a, b, c, d \in [x_1, x_2]$ such that $a < b \leq c < d$. Then,
\[
\frac{f(b) - f(a)}{b - a} \geq \frac{f(d) - f(c)}{d - c}
\]
\end{lemma}
\begin{proof}
Consider the intervals $[a, b]$ and $[c, d]$. By the mean value theorem (Theorem~\ref{theorem:mvt}), there exist $s_{a,b}$ and $s_{c,d}$ such that 
\[
s_{a,b} = \frac{f(b) - f(a)}{b - a}, s_{c,d} = \frac{f(d) - f(c)}{d - c},
\]
Since the function is concave on $[x_1, x_2]$, its derivative $f'$ is decreasing on this interval~\cite{thomas2014thomas}. Therefore since $s_{a,b}$ and $s_{c,d}$ are evaluated on points in $(a, b)$ and $(c, d)$ respectively, it follows that
\[
s_{a,b} \geq s_{c,d}
\]
as desired.
\end{proof}

\begin{lemma}
\label{lemma:pi4-interval}
Let $f(x) = \sin^2 x$, where $x \in [0, \pi/2]$. Let $x, x' \in [0, \pi/2]$, with $x' < x$ be such that $\sin^2 x - \sin^2 x' \leq 1/2$. Then $x - x ' \leq \pi/4$. 
\end{lemma}
\begin{proof}
We use the identity,
\[
\sin^2 x - \sin^2 x' = \sin(x + x') \sin(x-x')
\]
Clearly, $x + x' \geq x - x'$. Since $x \in [0, \pi/2]$, it follows that $x + x' \leq \pi -(x - x')$. Thus,
\[
x - x ' \leq x + x ' \leq \pi - (x - x')
\]
It follows that $\sin(x + x') \geq \sin(x - x')$, since the minimum of the sine function is at the end points $x - x'$ or $\pi - (x - x')$, both of which equal $\sin(x - x')$. Thus,
\begin{align*}
    \sin^2 x - \sin^2 x' &= \sin(x + x') \sin(x-x') \\
    &\geq \sin(x - x') \sin(x-x')\\
    &\geq \sin^2(x - x')
\end{align*}
This means
\begin{align*}
    \sin^2(x - x') &\leq \frac{1}{2} \\
    \Rightarrow \sin(x - x') &\leq \frac{1}{\sqrt{2}} \\
    \Rightarrow x - x' &\leq \sin^{-1}\left(\frac{1}{\sqrt{2}} \right) = \frac{\pi}{4}
\end{align*}
\end{proof}

\begin{lemma}
\label{lemma:outside-interval}
Let $x \in [\pi/4, \pi/2]$ and $x' \in [0, \pi/4]$ such that $x - x ' \leq \pi/4$. Then 
\[
f(x) - f(x') \geq f(x + \pi/4 - x') - f(\pi/4)
\]
\end{lemma}
\begin{proof}
Consider the intervals $[\pi/4, \pi/2 - x']$ and $[x, x + \pi/4 - x']$. See for example Figure~\ref{fig:sin2x}. We either have 
\[
\pi/4 < \pi/2 - x' \leq x < x + \pi/4 - x'
\]
or 
\[
\pi/4 < x < \pi/2 - x' < x + \pi/4 - x'
\]
Depending on the case, we can apply either Lemma~\ref{lemma:ave-slope} or Lemma~\ref{lemma:ave-slope-2} to obtain
\begin{align*}
  \frac{f(\pi/2 - x') - f(\pi/4)}{\pi/2 - x' - \pi/4} &\geq \frac{f(x + \pi/4 - x') - f(x)}{x + \pi/4 - x' - x} \\  
\Rightarrow \frac{f(\pi/2 - x') - f(\pi/4)}{\pi/4 - x'} &\geq \frac{f(x + \pi/4 - x') - f(x)}{\pi/4 - x'} \\ 
\Rightarrow f(\pi/2 - x') - f(\pi/4) &\geq f(x + \pi/4 - x') - f(x)
\end{align*}
Now
\begin{align*}
    f(x) - f(x') &= f(x) + f(\pi/4) - f(\pi/4) - f(x') \\
    &= f(x) - f(\pi/4) + f(\pi/4) - f(x') 
\end{align*}
Now, by the symmetry of the $\sin^2x$ function around $\pi/4$ (see Figure~\ref{fig:sin2x}),
\[
 f(\pi/4) - f(x') = f(\pi/4 + \pi/4 - x') - f(\pi/4) = f(\pi/2 - x') - f(\pi/4).
\]
Therefore,
\begin{align*}
    f(x) - f(x') &= f(x) - f(\pi/4) + f(\pi/2 - x') - f(\pi/4) \\
    & \geq f(x) - f(\pi/4) + f(x + \pi/4 - x') - f(x) \\
    &= f(x + \pi/4 - x') - f(\pi/4)
\end{align*}
\end{proof}

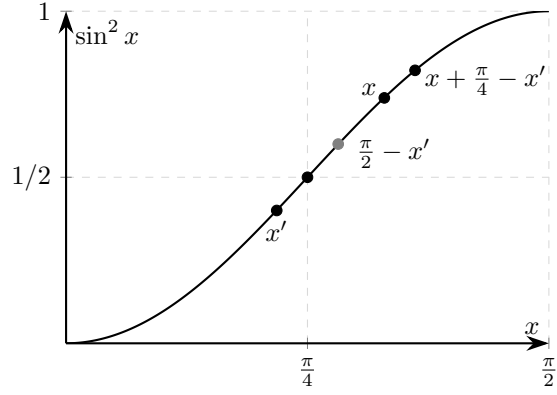
\begin{figure}
    \centering
\begin{tikzpicture}
  \begin{axis}[
    domain=0:pi/2,
    samples=200,
    xmin=0, xmax=1.57079632679,   % pi/2
    ymin=0, ymax=1,
    axis lines=middle,
    xlabel={$x$},
    ylabel={$\sin^2 x$},
    xtick={0, pi/4, pi/2},
    xticklabels={$0$, $\tfrac{\pi}{4}$, $\tfrac{\pi}{2}$},
    ytick={0,1/2,1},
    yticklabels={0, 1/2, 1},
    grid=both,
    grid style={dashed,gray!30},
    width=8cm,
    height=6cm,
    enlargelimits=false,
    clip=true,
    % Bigger arrow tips
    axis line style={thick, ->, >={Stealth[, length=3mm, width=2mm]}},
  ]
    % Use deg(x) so pgf math gets the argument in degrees (safe conversion)
    \addplot [thick,domain=0:pi/2] {sin(deg(x))^2};
    % Optional: mark endpoints
    \addplot+[only marks, mark=*, mark options={fill=black,draw=black}] coordinates {(pi/4 - 0.1, 1/2 - 0.1) (pi/4,1/2) (pi/4 + 0.25, 0.739) (pi/4 + 0.35, 0.822)} ;
    \addplot+[only marks,mark=*, mark options={fill=gray, draw=gray}] coordinates {(pi/4 + 0.1, 1/2 + 0.1)} ;
    %\addplot+[only marks,mark=*, fill=black, draw=black] coordinates {(pi/4 - 0.1, 1/2 - 0.1) (pi/4,1/2) (pi/4 + 0.25, 0.739) (pi/4 + 0.35, 0.822)} ;

%    \addplot[gray,dashed] coordinates {(pi/4,0) (pi/4,0.5)};
    
    \node at (axis cs:pi/4-0.1,0.28) [anchor=south] {$x'$};
    \node at (axis cs:pi/4+0.28,0.50) [anchor=south] {$\frac{\pi}{2} - x'$};
    \node at (axis cs:pi/4+0.2,0.72) [anchor=south] {$x$};
    \node at (axis cs:pi/4+0.58,0.72) [anchor=south] {$x + \frac{\pi}{4} - x'$};
    
  \end{axis}
\end{tikzpicture}    
    \caption{The function $f(x) = \sin^2 x$ over the interval $[0, \pi/2]$. The function is concave on the interval $[\pi/4, \pi/2]$. Further note that $\sin^2 (\pi/2 - x') - \sin^2(\pi/4) = \sin^2{\pi/4} - \sin^2{x'}$.}
    \label{fig:sin2x}
\end{figure}

\begin{theorem}
Let $n \geq 2$ be a positive integer. Let $\theta_\alpha, \theta_{\alpha'} \in [0, \pi/2]$. Let $0 < \beta \leq \pi/2$ be such that $\sin^2 \beta = \frac{1}{n}$. Then if 
\[
|\sin^2 \theta_\alpha - \sin^2 \theta_{\alpha'}| = \frac{1}{n}
\]
then 
\[
|\theta_\alpha - \theta_{\alpha'}| \leq \beta
\]
\end{theorem}
\begin{proof}
Let $f(x) = \sin^2 x$. We first assume that $\theta_\alpha, \theta_{\alpha'} \in [\pi/4, \pi/2]$.
Without loss of generality, assume that $\theta_\alpha > \theta_{\alpha'}$. Set $\theta_\alpha = \frac{\pi}{2}$ and $\theta_{\alpha'} = \frac{\pi}{2} - \beta$. Then $|\theta_\alpha - \theta_{\alpha'}| = \beta$, and 
\begin{align*}
    |\sin^2 \theta_\alpha - \sin^2 \theta_{\alpha'}| &= 
    \left|\sin^2 \left( \frac{\pi}{2}\right) - \sin^2 \left( \frac{\pi}{2} - \beta\right)\right| \\
    & = \left| 1 - \cos^2 \beta\right| \\
    &= \sin^2 \beta = \frac{1}{n}
\end{align*}

Consider any other interval $[\theta_{\alpha'},  \theta_{\alpha}]$. This interval cannot be contained in $[\pi/2 - \beta, \pi/2]$, because otherwise $|\sin^2 \theta_\alpha - \sin^2 \theta_{\alpha'}| < 1/n$. So consider first the case 
\[
\theta_\alpha < \theta_{\alpha'} \leq \pi/2 -\beta < \pi/2
\]
Since $\sin^2 x$ is a concave on the interval $[\pi/4, \pi/2]$ and differentiable on $(\pi/4, \pi/2)$, by Lemma~\ref{lemma:ave-slope-2} 
\begin{align*}
    \frac{f(\theta_\alpha) - f(\theta_{\alpha'})}{\theta_\alpha - \theta_{\alpha'}} &\geq \frac{f(\pi/2) - f(\pi/2 - \beta)}{\beta}  \\
\Rightarrow \frac{1}{n} \frac{1}{\theta_{\alpha} - \theta_{\alpha'}} &\geq \frac{1}{n} \frac{1}{\beta}\\
\Rightarrow \theta_{\alpha} - \theta_{\alpha'} &\leq \beta
\end{align*}
as required. Next consider the case 
\[
\theta_\alpha < \pi/2 < \theta_{\alpha'} < \pi/2 - \beta
\]
This time from Lemma~\ref{lemma:ave-slope}, we have
\begin{align*}
\frac{f(\theta_{\alpha}) - f(\theta_{\alpha'})}{\theta_{\alpha} - \theta_{\alpha'}} &\geq \frac{f(\pi/2) - f(\pi/2 - \beta)}{\beta} \\
\Rightarrow \theta_{\alpha} - \theta_{\alpha'} &\leq \beta
\end{align*}

Next consider $\theta_\alpha, \theta_{\alpha'} \in [0, \pi/4]$. Note that
\[
\sin^2 \left( \frac{\pi}{2} - x\right) = \cos^2 x\\
= 1 - \sin^2 x
\]
This means that if $\theta_\alpha, \theta_{\alpha'} \in [0, \pi/4]$ and without loss of generality, assuming $\theta_\alpha < \theta_{\alpha'}$ we get
\[
\sin^2 \left( \frac{\pi}{2} - \theta_\alpha\right) - \sin^2 \left( \frac{\pi}{2} - \theta_{\alpha'}\right)  =  1-  \sin^2 \theta_\alpha - 1 + \sin^2 \theta_{\alpha'}  =  \sin^2 \theta_{\alpha'} - \sin^2 \theta_{\alpha}
\]
Thus, the above analysis holds in this case as well by working with the angles $\frac{\pi}{2} - \theta_\alpha$ and $\frac{\pi}{2} - \theta_{\alpha'}$. 

The only case left is when one of the angles is in $[0, \pi/4]$ and the other in $[\pi/4, \pi/2]$. Again, without loss of generality, assume that $\theta_{\alpha'} \in [0, \pi/4]$ and $\theta_{\alpha} \in [\pi/4, \pi/2]$. By assumption $\sin^2 \theta_\alpha - \sin^2 \theta_{\alpha'} = \frac{1}{n}$. Since $n \geq 2$, $\sin^2 \theta_\alpha - \sin^2 \theta_{\alpha'}  \leq 1/2$, and from Lemma~\ref{lemma:pi4-interval} this implies that $\theta_{\alpha} -  \theta_{\alpha'} \leq \pi/4$. We can thus apply Lemma~\ref{lemma:outside-interval} to get
\[
\frac{f(\theta_\alpha) - f(\theta_{\alpha'})}{\theta_\alpha - \theta_{\alpha'}} \geq \frac{f(\theta_\alpha + \pi/4 - \theta_{\alpha'}) - f(\pi/4)}{\theta_\alpha - \theta_{\alpha'}}
\]
Depending on the case, either Lemma~\ref{lemma:ave-slope} or Lemma~\ref{lemma:ave-slope-2} implies
\begin{align*}
    \frac{f(\theta_\alpha) - f(\theta_{\alpha'})}{\theta_\alpha - \theta_{\alpha'}} &\geq \frac{f(\theta_\alpha + \pi/4 - \theta_{\alpha'}) - f(\pi/4)}{\theta_\alpha - \theta_{\alpha'}} \\
    &\geq \frac{f(\pi/2) - f(\pi/2 - \beta)}{\beta} \\
\Rightarrow  \frac{f(\theta_\alpha) - f(\theta_{\alpha'})}{\theta_\alpha - \theta_{\alpha'}} &\geq \frac{f(\pi/2) - f(\pi/2 - \beta)}{\beta} \\
\Rightarrow \theta_{\alpha} - \theta_{\alpha'} &\leq \beta
\end{align*}

This covers all cases except when $\theta_\alpha = \theta_{\alpha'}$. But in that case $|\sin^2 \theta_\alpha - \sin^2 \theta_{\alpha'}| = 0$. 
\end{proof}

Summarizing, we have

\begin{corollary}
\label{cor:sensitivity-angle}
Let $q$ be a counting query. Over all neighboring datasets $D$ and $D'$ with $q(D) = \alpha = \sin^2 \theta_\alpha$ and $q(D') = \alpha' = \sin^2 \theta_{\alpha'}$, where $0 \leq \theta_\alpha, \theta_{\alpha'} \leq \pi/2$, we have 
\[
|\theta_\alpha - \theta_{\alpha'}| \leq \sin^{-1}\left( \frac{1}{\sqrt{n}} \right)
\]
\qed
\end{corollary}

\subsubsection{Making the Mechanism Differentially Private}

Our first result is a straightforward application of differential privacy to the outcome $y$ from the amplitude estimation algorithm. 

\begin{corollary}
\label{cor:lap-amp-est}
Let 
\[
M < \frac{\pi}{\sin^{-1}(1 /\sqrt{n})}
\]
and let $y$ be the output of the amplitude estimation algorithm. Then adding Laplace noise of scale $\pi/M$ to the outcome $\pi y/M$ is $\epsilon$-DP.
\end{corollary}
\begin{proof}
From Corollary~\ref{cor:sensitivity-angle}, if we add 
\[
\text{Lap}\left(\frac{\sin^{-1}(1 /\sqrt{n})}{\epsilon}\right)
\]
to the angle $\theta_\alpha$ of the quantum amplitude estimation algorithm, then we achieve $\epsilon$-DP, since the sensitivity of the angle is $\sin^{-1}(1 /\sqrt{n})$. However, we note that the measured angles are encoded in steps of $\pi/M$. Thus, if $\pi/M$ is greater than $\sin^{-1}(1 /\sqrt{n})$, then the measured angle over two neighboring datasets can differ by at most $\pi/M$. Now
\[
  M < \frac{\pi}{\sin^{-1}(1 /\sqrt{n})} \Rightarrow \frac{\pi}{M} >  \sin^{-1}(1 /\sqrt{n}) 
\]
and therefore adding noise $\text{Lap}(\pi/M\epsilon)$ to $\pi y/M$ (before applying the $\sin^2$ function) is $\epsilon$-DP.
\end{proof}

For example, with $n = 10^6$, to apply the above mechanism we need $M < 3142$. However, the drawback of this mechanism is that we end up adding more noise since $\pi/M$ needs to be greater than the global sensitivity of the angle. Thus we need to add noise to the angle $\theta_\alpha$ itself before measurement. 

\descr{Alternative Mechanism.} To do so, we define the unitary operator 
\[
U_\eta \ket{\phi_\pm} = e^{i 2 \eta}\ket{\phi_\pm}
\]
where 
\[
\eta \sim \text{Lap}\left(\frac{\sin^{-1}(1 /\sqrt{n})}{\epsilon}\right)
\]
This operator can be realized by using the general definition of the operator $Q$ given in \cite{Brassard_2002} as
\[
Q = Q(A, q, \varphi_1, \varphi_2) = -AS_0(\varphi_1)A^{\dagger}S_\mathsf{G}(\varphi_2)
\]
where $S_\mathsf{G}(\varphi_2)$ is the generalization of $S_\mathsf{G}$ which maps
\[
\ket{x} \rightarrow \begin{cases}
    e^{i \varphi_2}\ket{x}, &\quad \text{ if } q(x) = 1 \\
    \ket{x}, &\quad \text{ if } q(x) = 0
\end{cases}
\]
Given this, we have $U_\eta := Q(A, q, 2\eta, -2\eta)$. Further note that applying the operator $\Lambda_M$ on $U_\eta$ has the effect:
\[
\Lambda_M(U_\eta):\;
\ket{j}\ket{x} \longmapsto \ket{j}U_\eta^{\,j}\ket{x}.
\]
Our idea to apply DP noise of scale $\eta$ to the angle $\theta_\alpha$ is to apply $\Lambda_M(U_\eta)$ after $\Lambda_M(Q)$. To see why this works, note that 
% It proves that this protocol is DP with respect to the classical data as well. Now it is remaining to add some noise from a $\text{Lap}(\pi/M(\epsilon-\epsilon_{dep}))$ distribution to the phase to ensure that the protocol has $\epsilon$ DP sensitivity. 
% To implement that, we are going to apply $\Lambda_M(U_\eta)$ after $\Lambda_M(Q)$.
% Define,
% \[
% U_\eta := Q(A,X,2\eta,-2\eta).
% \]
% Hence, 
% \has{I think we should define the eigenvalues as $e^{2i\eta}$. Then you can write them as $e^{2\pi i(\eta/\pi)}$}
% \[
% U_\eta \ket{\Psi_+} = e^{2 i \eta}\ket{\Psi_+}
% \]
% and,
% \[
% \Lambda_M(U_\eta):\;
% \ket{j}\ket{x} \longmapsto \ket{j}\bigl(U_\eta^{\,j}\ket{x}\bigr).
% \]
if we apply $\Lambda_M(U_\eta)$ to  $\ket{\phi_D^{(3)}}$ from Eq.~\eqref{eq:phiD_3}, we get (ignoring the $\ket{\phi_-}$ term which has a similar effect)

\begin{align*}
\Lambda_M(U_\eta)
\left(
  \frac{1}{\sqrt{M}}
  \sum_{j=0}^{M-1}
    e^{i2\pi  j \frac{\theta_\alpha}{\pi}}
    \ket{j}\otimes\ket{\phi_+}
\right)
&=
\frac{1}{\sqrt{M}}
\sum_{j=0}^{M-1}
  e^{i2\pi j \frac{\theta_\alpha}{\pi}}
  \ket{j}\otimes U_\eta^{\,j}\ket{\phi_+} \\[4pt]
&=
\frac{1}{\sqrt{M}}
\sum_{j=0}^{M-1}
  e^{i2\pi j \frac{\theta_\alpha}{\pi}}
  \ket{j}\otimes
  e^{i2\pi j \frac{\eta}{\pi}}\ket{\phi_+} \\
&=
\frac{1}{\sqrt{M}}
\sum_{j=0}^{M-1}
  e^{i2\pi  j\left(\frac{\theta_\alpha}{\pi}+\frac{\eta}{\pi}\right)}
  \ket{j}\otimes\ket{\phi_+} \\
&=
\ket{S_M\Big(\frac{\theta_\alpha+\eta}{\pi}\Big)}\otimes\ket{\phi_+}
\end{align*}
The action on $\ket{\phi_-}$ is similar. Therefore, the state 
before applying $\text{QFT}_M^\dagger$ is
\[
\ket{\phi_D^{(4)}}=\frac{e^{i\theta_{\alpha}}}{\sqrt{2}}\;\ket{S_M(\tfrac{\theta_{\alpha}+\eta}{\pi})}\otimes\ket{\Psi_+}
\;-\;
\frac{e^{-i\theta_{\alpha}}}{\sqrt{2}}\;\ket{S_M(1-\tfrac{\theta_{\alpha}+\eta}{\pi})}\otimes\ket{\Psi_-}.
\]
After applying $\text{QFT}_M^\dagger$ and measuring in the computational basis, the measured value $y$ is differentially private as $\eta$ is scaled according to the global sensitivity of the angle $\theta_\alpha$.

\subsection{Inherent Noise Due to the Depolarizing Channel} 

The quantum circuits acting on the state $\phi_D$ are inherently noisy. It is natural to expect that this noise may provide some uncertainty in the outcome leading naturally to some sort of differentially private guarantee. Researchers have indeed looked into this and have proven differential privacy properties of noisy quantum channels \cite{angrisani2022differential, zhou2017differential, yuxuan-quantum-adversarial-attacks, guan2025optimalmechanismsquantumlocal}. Here we take an example of the depolarizing channel. In particular we assume that in between any unitary used in the quantum circuit to compute $q(D)$, depolarizing noise is applied. Now using some properties of the depolarizing channel, as described below, we can factor this into the differentially private guarantees of the two differentially private algorithms described earlier. Depolarizing noise commonly arises in practical quantum devices, particularly in near term realizable devices.

\begin{definition}[Depolarizing channel]
Given a state $\rho$ in an $m$-dimensional Hilbert space, the depolarizing channel with probability parameter $p$ is defined as
\[
\mathcal{D}_p(\rho) \;=\; (1-p)\,\rho \;+\; p\,\frac{I}{m}, \qquad p\in[0,1].
\]
\qed
\end{definition}
If for all inputs $\rho$ and $\sigma$ we have $\mathcal{T}(\rho, \sigma) \leq \tau$, then the mechanism $\mathcal{D}_p(\rho)$ satisfies $(\tau, \epsilon, 0)$-quantum differential privacy~\cite{zhou2017differential}, with 
\begin{equation}
\label{eq:eps-depolar}
    \epsilon = \ln\left( 1 + \frac{1 - p}{p}\, m \tau \right)
\end{equation}
Furthermore, for all neighboring datasets $D$ and $D'$ with quantum encoding $\rho_D$ and $\rho_{D'}$, respectively, we know from Section~\ref{subsec:dp-amplify} that 
\[
\mathcal{T}(\rho_D, \rho_{D'}) \leq \sqrt{1 - \hat\kappa_\phi}
\]
Thus, $\mathcal{D}_p(\rho_D)$ satisfies $(\sqrt{1 - \hat\kappa_\phi}, \epsilon, 0)$-quantum differential privacy. It follows from Lemma 1 in~\cite{angrisani2022differential} that this mechanism is $(\epsilon, 0)$-differentially private. 

With these results, we can utilize this channel as follows. Let us assume that $\epsilon$ is our overall budget. Let $(\epsilon_1, \delta)$ be the differentially private guarantee from the Laplace mechanism (through either of the two algorithms defined earlier), and $\epsilon_2$ be the guarantee through the depolarizing channel. Then the mechanism as a whole, with explicit Laplace noise and implicit depolarization noise, is $ (\epsilon_1, \delta) + (\epsilon_2, 0) = (\epsilon_1 + \epsilon_2, \delta) := (\epsilon, \delta)$-differentially private through sequential composition of differential privacy (see Theorem~\ref{prop:composition}).

The only thing to consider is that depolarizing noise is applied after every unitary. It turns out, using the properties of the depolarizing channel as shown below, we only need to account for an upper bound on the total probability parameter $p_\text{tot}$ of the composition of all depolarizing channels. Specifically, define
\[
p_{\mathrm{tot}} \;:=\; 1- \prod_{i=1}^{\ell} (1-p_i),
\]
The following are known results about the depolarizing channel.

\begin{lemma}[Unitary covariance]
For any unitary $U$, $\mathcal{D}_p(U\rho U^\dagger) \;=\; U\,\mathcal{D}_p(\rho)\,U^\dagger.
\quad$ 
\end{lemma}
\begin{proof}
\[
\mathcal{D}_p(U\rho U^\dagger)=(1-p)U\rho U^\dagger + p\,\frac{I}{d}
=U\big((1-p)\rho+p\,\tfrac{I}{d}\big)U^\dagger
=U\,\mathcal{D}_p(\rho)\,U^\dagger.
\]
\end{proof}
\begin{lemma}[Composition]\label{lem:composition}
For any $p_1,p_2\in[0,1]$,
\[
\mathcal{D}_{p_2}\circ \mathcal{D}_{p_1} \;=\; \mathcal{D}_{\,p_1+p_2-p_1p_2}
\;=\; \mathcal{D}_{\,1-(1-p_1)(1-p_2)}.
\]
And generalizing it for any $\ell \ge 2$,
\[
\mathcal{D}_{p_\ell}\circ\cdots\circ \mathcal{D}_{p_1}
\;=\;
\mathcal{D}_{\,1-\prod_{i=1}^{\ell}(1-p_i)}.
\]
\end{lemma}
\begin{proof}
\begin{align*}
\mathcal{D}_{p_2}(\mathcal{D}_{p_1}(\rho))
&=\mathcal{D}_{p_2}\big((1-p_1)\rho + p_1\tfrac{I}{d}\big) \\
&=(1-p_2)\big((1-p_1)\rho + p_1\tfrac{I}{d}\big) + p_2\tfrac{I}{d} \\
&=(1-p_1)(1-p_2)\rho + \big[p_1(1-p_2)+p_2\big]\tfrac{I}{d}.
\end{align*}
Let $p_\ast := p_1+p_2-p_1p_2 = 1-(1-p_1)(1-p_2)$. Then the last line equals
$(1-p_\ast)\rho + p_\ast \tfrac{I}{d} = \mathcal{D}_{p_\ast}(\rho)$.

For $\ell>2$, iterate the two-step identity. After $j$ steps the coefficient of $\rho$ is
$\prod_{i=1}^{j}(1-p_i)$; thus after $\ell$ steps the channel is
$(1-p_{\mathrm{tot}})\rho + p_{\mathrm{tot}}\tfrac{I}{d}$ with
$p_{\mathrm{tot}} = 1-\prod_{i=1}^{m}(1-p_i)$.
\end{proof}

Now combining these two lemmas we get,
\begin{theorem}\label{thm:overall}
Let $U_1,\dots,U_\ell$ be arbitrary unitaries on same $m$-dimensional register and
$\mathcal{D}_{p_1},\dots,\mathcal{D}_{p_\ell}$ be depolarizing channels with specified probability parameters. Then
\[
\mathcal{D}_{p_\ell}\!\circ\!\mathcal{A}d_{U_\ell}\!\circ\cdots\circ\mathcal{D}_{p_1}\!\circ\!\mathcal{A}d_{U_1}
\;=\;
\mathcal{A}d_{U_\ell \cdots U_1}\!\circ\!\mathcal{D}_{\,1-\prod_{i=1}^{\ell}(1-p_i)}.
\]
where, $\mathcal{A}d_U(\rho) := U\rho U^\dagger$.
\end{theorem}
This theorem ensures that, regardless of where the depolarizing noise occurs, if we know the total depolarization bound or more precisely, an upper bound on $p_{\text{tot}}$, we can use it to obtain $\epsilon_2$ by inserting $p = p_\text{tot}$ and $\tau = \sqrt{1 - \hat\kappa_\phi}$ into Eq.~\eqref{eq:eps-depolar}.

\section{Applying Homomorphic Encryption to Delegate Computation}
\label{sec:homomorphic}
To hide the contents of the dataset and the results of the queries from the server $\server$, we can use the quantum one-time pad (QOTP). For both methods for finding the query answer, we can encrypt the initial state using the quantum one-time pad (QOTP), as 
\[
X^{\otimes a}Z^{\otimes b}\ket{\phi_D}
\]
where $a$ and $b$ are random $m$-bit strings known as the \emph{keys}. Given a universal gate set such as 
\[
\{X, Z, \text{CNOT}, H, S, T\},
\]
we can implement any quantum circuit using these gates. The Pauli gates $X$ and $Z$ commute or anti-commute with the encrypted state above, and hence we can apply them as if we are applying them directly to the unencrypted state (with a possible inconsequential change in global phase). On the other hand $H$, $S$ and $\text{CNOT}$ are Clifford gates, and they too can be applied to the encrypted state with a possible update to the keys $a$ and $b$. The final gate $T$ can be implemented using so-called magic states. These are applied to a pre-computed state independent of $\ket{\phi_D}$. 

As mentioned in Section~\ref{subsec:pred-queries}, a predicate query can be implemented using a combination of $X$, $\text{CNOT}$ and multi-control Toffoli gates. A multi-control Toffoli gate can be implemented as a series of Toffoli gates, which themselves can be implemented using $H$, $T$, \text{CNOT} and $S$ gates~\cite[\S 4.3]{nielsen-chuang}. Thus, the entire query circuit can be done homormorphically, with key updates as required by the gates applied. Since the query itself is public, the key updates can be done by the client itself. For completeness, we show this in Appendix~\ref{app:hom}. 

Depending on the query, and the method to obtain query answers, such homomorphic computations will involve significant client involvement. Of particular interest is the case when a query can be answered only via Pauli and Clifford gates. We see that if the dataset is binary, and we are interested in a predicate query of a single attribute, e.g., if the $i$th binary attribute is 0 or 1, this can be done using only Clifford gates if we are to apply the direct measurement method. For the rest of the queries and datasets, we require the use of Toffoli gates, which inadvertently need magic states to homomorphically implement the $T$ gate. 

\section{Related Work} 
The framework of differential privacy was proposed by Dwork et al. \cite{dwork2006calibrating, dwork-dp-book}. Counting queries and its subset, predicate queries, are some of the basic statistics desired from a dataset, and therefore, unsurprisingly, they have been studied in depth in the differential privacy literature~\cite{ullman2013answering, mckenna2021hdmm, smith2022making}. 
%\textbf{Classical Differential Privacy. } In the classical realm,  proposed the framework of differential privacy. The concept of answering the counting queries, such as predicate queries, plays a crucial role in the formalism of classical DP. This aspect of DP has been studied numerous times clasically, for example,.
The concept of quantum differential privacy (QDP) was first proposed by Zhou and Ying~\cite{zhou2017differential} who defined QDP using trace distance and quantum operations. The work in~\cite{hirche-qdp-inf-theory} gave an information-theoretic formulation of QDP. Our work is more related to recent works that connect classical differential privacy to its quantum variant, including deriving differentially private properties of quantum encoded datasets~\cite{aaronson2019gentle, angrisani2023unifyingframeworkdifferentiallyprivate, angrisani-ldp, angrisani2022differential, zhao-bridging,guan2025optimalmechanismsquantumlocal, asghar2026efficient}. 
The key motivation of our work comes from two main results in~\cite{angrisani2022differential}. The first result is where they show that a quantum algorithm that only accesses a quantum encoded dataset as input possesses some differential privacy guarantees. This motivated us to analyze differential privacy properties of algorithms that compute queries on a quantum encoded dataset. The other result from~\cite{angrisani2022differential} is the repeated measurement differentially private algorithm which adds Laplace noise to the final result. We refine the analysis of this algorithm for counting queries and show that privacy is in fact amplified. Our use of depolarizing noise is motivated by the work of Du et al.~\cite{yuxuan-quantum-adversarial-attacks} who show that inherent depolarizing noise in quantum channels can be utilized to make a machine learning algorithm robust to adversarial inputs: imperceptible changes made to inputs that nevertheless induce misclassifications. 

Finally, our second mechanism uses the amplitude estimation algorithm from Brassard et al.\cite{Brassard_2002}. Their algorithm uses quantum phase estimation to approximate the amplitude of a target state. Due to the reliance on phase estimation, it is not suitable for Noisy Intermediate-Scale Quantum (NISQ) devices. To address this limitation, two main alternatives have been developed in the literature: the maximum likelihood based amplitude estimation approach~\cite{suzuki2020amplitude, callison2023improvedmaximumlikelihoodquantumamplitude, giurgica2022low, uno2021modified} and the iterative amplitude estimation approach \cite{grinko2021iterative, aaronson2020quantum}. 
%More recently, the Adaptive Windowed Quantum Amplitude Estimation (AWQAE) framework introduced a highly scalable, modular approach that decouples estimation precision from qubit counts via iterative, chunk-based phase estimation and classical ambiguity correction\cite{shukla2025modularquantumamplitudeestimation}. 
However, we use the canonical quantum amplitude estimation algorithm in our work since it requires  limited number of initial state preparations, and the fact that our focus is not on NISQ devices.  

\section{Conclusion}
We have shown how to answer counting queries with differential privacy through two canonical methods of amplitude estimation with differential privacy. Our results show refined privacy properties of these techniques for the specific case of counting queries. Several aspects of our work may be worth exploring in future work. One such aspect is looking at other classes of queries which may enjoy some privacy properties that are not present in the classical setting.

%Compared with classical delegated approaches, our framework offers two main advantages. First, quantum encoding itself amplifies privacy by making neighboring datasets harder to distinguish. Second, amplitude estimation provides a quadratic query advantage over direct sampling in high-precision settings. In delegated computation, QOTP-based encrypted evaluation also offers a promising alternative to classical FHE, particularly for circuits dominated by Pauli and Clifford operations. Overall, quantum encoding serves not just as a data representation, but as a resource for both privacy and computational gain.

%At the same time, the limitations of the present work should be emphasized clearly. 
A drawback of the repeated measurement method for answering counting queries is that it requires multiple state preparations. An interesting question is whether we can answer multiple queries on a single state, even when it is altered after measurement. The amplitude estimation approach, while more query-efficient, is tied to canonical phase estimation based machinery and therefore assumes an error-corrected and well-controlled quantum setting rather than NISQ hardware. 
%.

%These limitations point directly to several future directions. 

%Being said that, a priority is to extend the analysis beyond predicate queries to broader classes of statistical queries and to more general quantum data encodings. 

Another intriguing question is whether there are types of queries that can be answered via only Pauli or Clifford gates. Such queries can then be outsourced to a server who can compute them blindly and non-interactively. As we discuss in this paper, such delegated homomorphic execution is generally not free, in the sense that circuits for non-trivial predicate queries require Toffoli gates. These gates inevitably use non-Clifford gates, which only satisfy homomorphic properties through magic-state preparation resulting in non-negligible client interaction for key updates. We have shown that for a binary dataset and a predicate query on a single attribute, this is indeed possible with Clifford gates.   

\bibliographystyle{plain}
\bibliography{quantum-queries}

\appendix 

\section{Homomorphic Operations in the Quantum One-Time Pad}
\label{app:hom}
Assume we are given the state $\ket{\phi_D}$ from the Hilbert space $\mathcal{H}^{\otimes m}$. Let $a, b \in \{0, 1\}^m$ be $m$-bit strings. Denote by $a_i$, respectively $b_i$, the $i$th bit of $a$, respectively $b$. The state $\ket{\phi_D}$ is encryted as
\begin{equation}
\label{eq:encrypted-state}
X^aZ^b\ket{\phi_D} = (X^{a_1} \otimes \cdots \otimes X^{a_m})(Z^{b_1} \otimes \cdots \otimes Z^{b_m}) \ket{\phi_D}
\end{equation}
We will show how we can apply various Pauli, Clifford and non-Clifford gates homomorphically. The latter requires special circuitry and adaptive measurements. These notes are detailed from~\cite[\S 13]{vidick2023introduction}

\paragraph{The $X$ and $Z$ Gates.} Suppose we wish to apply $X$ to the $i$th qubit of $\ket{\phi_D}$. That is:
\[
X^{e_i}\ket{\phi_D} = (I \otimes \cdots \otimes X \otimes \cdots \otimes I) \ket{\phi_D}
\]
where the $X$ appears is the $i$th position, as indicated by the use of the symbol $e_i$ as a superscript to $X$. On the encrypted state, this only affects the operators on the $i$th qubit. Thus we get:
\[
X X^{a_i}Z^{b_i} \ket{\phi_D} = X^{a_i} XZ^{b_i} \ket{\phi_D}= (-1)^{b_i} X^{a_i}Z^{b_i} X  \ket{\phi_D}
\]
where we have used the fact that $X$ anti-commutes with $Z$. Therefore, with the exception of an inconsequential global phase, this is as if $X$ was applied on the initial state. This example easily extends to the $Z$ gate.

\paragraph{The $S$ and $H$ Gates.} Assume we wish to apply the $S$ gate on the $i$th qubit. As before this only changes the $i$th qubit. So we only focus on the $i$th encrypted qubit. Recall the identities~\cite[\S 10.5.2]{nielsen-chuang}:
\[
SXS^\dagger = Y = iXZ, \; SZS^\dagger = Z
\]
We have
\begin{align*}
 S X^{a_i}Z^{b_i} \ket{\phi_D} &= S X^{a_i} S^\dagger S Z^{b_i} \ket{\phi_D} = Y^{a_i} S Z^{b_i}\ket{\phi_D} \\
 &= Y^{a_i} S Z^{b_i} S^\dagger S \ket{\phi_D} = Y^{a_i} Z^{b_i} S\ket{\phi_D} \\
 &= (iXZ)^{a_i} Z^{b_i} S\ket{\phi_D} = i^{a_i} X^{a_i} Z^{a_i \oplus b_i} S \ket{\phi_D}
\end{align*}
Thus, the client can update the $i$th bits of the keys $a$ and $b$ as :
\[
a_i \leftarrow a_i, \; b_i \leftarrow a_i \oplus b_i
\]
Again apart from the unimportant global phase $i^{a_i}$, this is as if the $S$ gate is directly applied to the original state.

For the Hadamard gate $H$, recall the identities~\cite[\S 10.5.2]{nielsen-chuang}:
\[
HXH^\dagger = Z, \; HZH^\dagger = X
\] 
Thus applying the $H$ gate to the $i$th qubit yields
\begin{align*}
 H X^{a_i}Z^{b_i} \ket{\phi_D} &= H X^{a_i} H^\dagger H Z^{b_i} \ket{\phi_D} = Z^{a_i} H Z^{b_i}\ket{\phi_D} \\
 &= Z^{a_i} H Z^{b_i} H^\dagger H \ket{\phi_D}  = Z^{a_i} X^{b_i} H\ket{\phi_D} \\
 &= (-1)^{a_ib_i} X^{b_i} Z^{a_i } H \ket{\phi_D}
\end{align*}
The clients updates the $i$th bits of the keys $a$ and $b$ by swapping $a_i$ and $b_i$.  

\paragraph{The CNOT Gate.} For simplicity of notation, let $U$ denote the CNOT gate. Recall the following identities~\cite[\S 10.5.2]{nielsen-chuang}:

\begin{align*}
    UX_1U^\dagger &= X_1X_2, \\ UX_2U^\dagger &= X_2 \\
    UZ_1U^\dagger &= Z_1, \\
    UZ_2U^\dagger &= Z_1Z_2 \\
\end{align*}
where the subscripts $1$ and $2$ represent the control and target qubits respectively. Let us assume that the control qubit is the $i$th qubit and the target qubit is the $j$th qubit. This implies
\begin{align*}
  U(X^{a_i}Z^{b_i} \otimes X^{a_j}Z^{b_j}) &= U(X^{a_i} \otimes X^{a_j})(Z^{b_i} \otimes Z^{b_j})  \\
  &= U(X^{a_i} \otimes I)(I \otimes X^{a_j})(Z^{b_i} \otimes Z^{b_j}) \\
  &= U(X^{a_i} \otimes I)U^\dagger U (I \otimes X^{a_j})(Z^{b_i} \otimes Z^{b_j}) \\
  &= (X^{a_i} \otimes X^{a_i}) U (I \otimes X^{a_j})(Z^{b_i} \otimes Z^{b_j}) \\
  &= (X^{a_i} \otimes X^{a_i}) U (I \otimes X^{a_j})U^\dagger U(Z^{b_i} \otimes Z^{b_j}) \\
  &= (X^{a_i} \otimes X^{a_i})(I \otimes X^{a_j})U(Z^{b_i} \otimes Z^{b_j}) \\
  &= (X^{a_i} \otimes X^{a_i \oplus a_j})U(Z^{b_i} \otimes I) (I \otimes Z^{b_j})\\
  &= (X^{a_i} \otimes X^{a_i \oplus a_j})U (Z^{b_i} \otimes I) U^\dagger U (I \otimes Z^{b_j}) U^\dagger U\\
  &= (X^{a_i} \otimes X^{a_i \oplus a_j}) (Z^{b_i} \otimes I) (Z^{b_j}  \otimes Z^{b_j}) U\\
  &= (X^{a_i} \otimes X^{a_i \oplus a_j}) (Z^{b_i \oplus b_j} \otimes Z^{b_j}) U\\
  &= (X^{a_i} Z^{b_i \oplus b_j} \otimes X^{a_i \oplus a_j} Z^{ b_j})U
\end{align*}
Thus the client updates the $i$th and $j$th bits of the keys as 
\[
a_i \leftarrow a_i,\; b_i \leftarrow b_i \oplus b_j,\; a_j \leftarrow a_i \oplus a_j,\; b_j \leftarrow b_j
\]
and otherwise the CNOT operator is shifted to the right. 

\paragraph{The $T$ Gate.} Since $T$ is a non-Clifford gate, it is handled differently. The $T$ gate is given by the matrix
\[
T = \begin{bmatrix}
    1 & 0 \\
    0 & e^{i\pi/4}
\end{bmatrix}
\]
This is done by first preparing the so-called magic state defined as:
\[
\ket{\pi/4} = T \ket{+} = \frac{1}{\sqrt{2}}(\ket{0} + e^{i\pi/4} \ket{1} )
\]
We also make use of the gate $S$ defined as
\[
S = \begin{bmatrix}
    1 & 0 \\
    0 & i
\end{bmatrix}
\]
with the following identities that have already been shown
\[
SXS^\dagger = Y = iXZ, \; SZS^\dagger = Z
\]
With this for any $a, b \in \{0, 1\}$, we get 
\begin{align*}
  SX^{a}Z^b  = i^a X^a Z^{a \oplus b} S
\end{align*}
as already shown. Once the magic state is prepared, the following circuit applies the $T$ gate to $\psi$. 

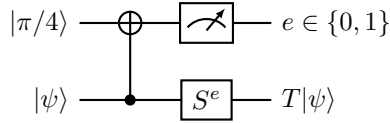
\begin{center}
\begin{quantikz}
\lstick{\ket{\pi/4}} &  \targ{} & \meter{} & \rstick{$e \in \{0, 1\}$} \\ 
\lstick{\ket{\psi}} &  \ctrl{-1} & \gate{S^e} & \rstick{$T$\ket{\psi}} \\
\end{quantikz}
\end{center}

To see why the result is correct, assume $\psi = \alpha \ket{0} + \beta \ket{1}$. Then, the first part of the circuit computes
\begin{align*}
 CX(\ket{\pi/4} \otimes \ket{\psi}) &=   \frac{1}{\sqrt{2}} CX\left(  \alpha \ket{0}\ket{0} + \beta  \ket{0}\ket{1} + \alpha e^{i\pi/4}\ket{1}\ket{0} + \beta e^{i\pi/4}  \ket{1}\ket{1}\right)\\
 &= \frac{1}{\sqrt{2}} \left(  \alpha \ket{0}\ket{0} + \beta  \ket{1}\ket{1} + \alpha e^{i\pi/4}\ket{1}\ket{0} + \beta e^{i\pi/4}  \ket{0}\ket{1}\right) \\
 &=  \frac{1}{\sqrt{2}} \ket{0} \otimes ( \alpha \ket{0} + \beta e^{i\pi/4}  \ket{1}) + \frac{1}{\sqrt{2}} \ket{1} \otimes ( \alpha e^{i\pi/4} \ket{0} + \beta   \ket{1}) \\
 &= \frac{1}{\sqrt{2}} \ket{0} \otimes T\ket{\psi} + \frac{1}{\sqrt{2}} \ket{1} \otimes T' \ket{\psi}
\end{align*}
where we define
\[
T' = \begin{bmatrix}
    e^{i\pi/4} & 0 \\
    0 & 1
\end{bmatrix}
\]
Now it is easy to see that $ST' = e^{i \pi/4} T$, which implies that 
\[
T' = e^{i \pi/4} S^\dagger T
\]
Thus, the above state becomes
\[
\frac{1}{\sqrt{2}} \ket{0} \otimes T\ket{\psi} + \frac{e^{i \pi/4}}{\sqrt{2}} \ket{1} \otimes S^\dagger T \ket{\psi}
\]
If now we measure the first qubit and the outcome is $e = 0$, then ignoring any global phase, the output is $T\ket{\psi}$. On the other hand, if the outcome is $e = 1$, then again ignoring global phase, the output is $S^e S^\dagger T \ket{\psi} = S S^\dagger T \ket{\psi} = T \ket{\psi}$, as desired. 

But the main point of the circuit is to apply it to an encrypted state. Let us assume that $\ket{\psi}$ is one of the encrypted qubits on which the $T$ gate needs to be applied. Let its encryption be $X^aZ^b \ket{\psi}$. Here $a$ and $b$ are random bits. We also encrypt the magic state $\ket{\pi/4}$ as $X^cZ^d \ket{\pi/4}$, with random bits $c$ and $d$. Let $U$ denote the CNOT gate with control being the second qubit and target being the first qubit. We have
\begin{align*}
    U(X^cZ^d \ket{\pi/4} \otimes X^aZ^b \ket{\psi}) &= U(X^cZ^d  \otimes X^aZ^b )(\ket{\pi/4} \otimes \ket{\psi}) \\
    &= U(X^c \otimes I)(Z^d \otimes I)(I \otimes X^a)(I \otimes Z^b)(\ket{\pi/4} \otimes \ket{\psi})\\
    &= (X^c \otimes I)U(Z^d \otimes I)(I \otimes X^a)(I \otimes Z^b)(\ket{\pi/4} \otimes \ket{\psi})\\
    &= (X^c \otimes I)(Z^d \otimes Z^d)U(I \otimes X^a)(I \otimes Z^b)(\ket{\pi/4} \otimes \ket{\psi})\\
    &= (X^c \otimes I)(Z^d \otimes Z^d)(X^a \otimes X^a)(I \otimes Z^b)U(\ket{\pi/4} \otimes \ket{\psi})\\
    &= (X^{a\oplus c} Z^d \otimes X^a Z^{b \oplus d} ) U (\ket{\pi/4} \otimes \ket{\psi})
\end{align*}
where we have used the conjugation identities for the CNOT operator as mentioned above. Now, through the analysis of the unencrypted case for the right hand term, and ignoring the global phase terms we get
\begin{align*}
    & (X^{a\oplus c} Z^d \otimes X^a Z^{b \oplus d} ) (\ket{0} \otimes T\ket{\psi} +  \ket{1} \otimes S^\dagger T \ket{\psi}) \\
    &= X^{a\oplus c} Z^d \ket{0} \otimes X^a Z^{b \oplus d} T\ket{\psi} \\
    &+ X^{a\oplus c} Z^d \ket{1} \otimes X^a Z^{b \oplus d} S^\dagger T \ket{\psi}
\end{align*}
Now first assume that $a \oplus c = 0$. If the measurement outcome is $e = 0$, then we get $X^a Z^{b \oplus d} T\ket{\psi}$ as the outcome. Thus this is the application of $T$ on the actual qubit, with the client having to update keys accordingly. On the other hand if $e = 1$, then we apply the $S$ gate to the outcome $X^a Z^{b \oplus d} S^\dagger T \ket{\psi}$ giving us
\[
S X^a Z^{b \oplus d} S^\dagger T \ket{\psi} = i^a X^a Z^{a \oplus b \oplus d} S S^\dagger T \ket{\psi}= i^a X^a Z^{a \oplus b \oplus d}  T \ket{\psi}
\]
Ignoring the global phase, this applies $T$ gate to the original qubit, with the client having to update the encryption keys accordingly. 

Next assume that $a \oplus c = 1$. Now, if the measurement outcome is $e = 0$, we get the output 
\[
X^a Z^{b \oplus d} S^\dagger T \ket{\psi} = (-i)^a S X^a Z^{a \oplus b \oplus d} T \ket{\psi},
\]
using the identity above to move $S$ to the other side. On the other hand if the measurement outcome is $e = 1$, then we get
\[
S X^a Z^{b \oplus d}  T \ket{\psi} 
\]
Thus, in both cases, i.e., $a \oplus c = 0$ or $a \oplus c = 1$, ignoring the global phase, the outcome is
\[
S^{a \oplus c} X^{a'} Z^{b'} T \ket{\psi}
\]
where $a'$ and $b'$ are client updates. This introduces the gate $S$ at the beginning if $a \oplus c = 1$. But this can be removed by the client requesting the server to apply $(S^{a \oplus c})^\dagger$. Crucially, this does not reveal information about the key $a$ which encrypts the data, since $a \oplus c = 1$ implies either $a = 0$ or $a = 1$ with equal probability, as $c$ is a random bit. 

\section{Complexity Analysis of the Two Methods to Estimate Query Answers}
\label{app:comp-analysis}
Let us analysis the complexities of two methods, i,e. the direct measurement method and the AE method. Our analysis is in terms of query complexity which is the indication of how many times we need to call the oracle to get a result within the error bound $\varepsilon$. Let us first discuss about the direct measurement method. Here, a single query (measurement) is a Bernoulli trial $X_i$ with success probability 
$\alpha$:$$X_i = \begin{cases} 1 & \text{with probability } \alpha \\ 0 & \text{with probability } 1-\alpha \end{cases}$$
The estimate of $\alpha$ with $t$ samples is 
$$\widetilde{\alpha} = \frac{1}{t} \sum_{i=1}^{t} X_i$$
The precision of this estimate is determined by the standard error (the standard deviation of the sampling distribution). Since the $t$ queries are independent, the variance of the average is
\[
\text{Var}(\widetilde{\alpha}) = \frac{\sum_{i=1}^t\text{Var}(X_i)}{t^2} = \frac{\alpha(1-\alpha)}{t}
\]
Thus, the standard deviation is
\[
\sqrt{\text{Var}(\widetilde{\alpha})} = \sqrt{\frac{\alpha(1-\alpha)}{t}}
\]
Hence, to achieve a target precision $\Delta$ we must have
$$\Delta = \sqrt{\frac{\alpha(1-\alpha)}{t}} \implies \Delta^2 = \frac{\alpha(1-\alpha)}{t} \implies t = \frac{\alpha(1-\alpha)}{\Delta^2}$$

Since $\alpha(1-\alpha)$ is a constant we can say that 

$$ t =  \mathcal{O}(1/\Delta^2)$$

On the other side for the amplitude estimation method, 
$$|\alpha - \alpha'| \leq \frac{2\pi\sqrt{\alpha(1-\alpha)}}{M}+\frac{\pi^2}{M^2}=\Delta,$$
which implies, $$M\in \mathcal{O}(1/\Delta) $$

\end{document}